\documentclass[sigconf,nonacm]{acmart}

\usepackage{booktabs}
\usepackage{multirow}
\usepackage{tabularx}
\usepackage{amsmath}
\usepackage[ruled,vlined,linesnumbered]{algorithm2e}
\usepackage{algpseudocode}
\usepackage{subcaption}
\usepackage{xcolor}

\usepackage{listings}
\usepackage{xcolor}

\definecolor{etalePurple}{HTML}{4B2E83}
\definecolor{etaleGold}{HTML}{85754D}
\definecolor{etaleComment}{HTML}{6B6B6B}
\definecolor{etaleBg}{HTML}{FBFAFD}

\lstdefinestyle{etalecpp}{
  language=C++,
  basicstyle=\ttfamily\footnotesize,
  keywordstyle=\color{etalePurple}\bfseries,
  commentstyle=\color{etaleComment}\itshape,
  stringstyle=\color{etaleGold},
  numberstyle=\tiny\color{etaleComment},
  backgroundcolor=\color{etaleBg},
  frame=single,
  rulecolor=\color{etalePurple},
  framesep=6pt,
  breaklines=true,
  showstringspaces=false,
  columns=fullflexible,
  keepspaces=true,
  morekeywords={uint32_t,uint64_t,int,float,bool,constexpr,struct,using,size_t,volatile,inline,return,const}
}

\newtheorem{proposition}{Proposition}
\newtheorem{lemma}{Lemma}
\theoremstyle{definition}

\newtheorem{corollary}{Corollary}
\theoremstyle{remark}
\newtheorem{remark}{Remark}

\usepackage{pifont}
\newcommand{\cmark}{\ding{51}}
\newcommand{\xmark}{\ding{55}}

\begin{document}

\title{ETALE: Evolving Topology with Accelerated Lock-free Execution for Dynamic Graph ANN Search on GPUs}

\author{Dongfang Zhao}
\affiliation{%
  \department{Tacoma School of Engineering and Technology}
  \institution{University of Washington}
  \city{}
  \country{}
}
\email{dzhao@uw.edu}


\begin{abstract}
  Graph-based approximate nearest neighbor (ANN) indexes on the GPU are originally built for static collections and must reconstruct the affected window to absorb any update, while the dynamic graph indexes that update incrementally run on the CPU and cannot exploit GPU parallelism.
  This paper presents a new ANN index, namely ETALE (Evolving Topology with Accelerated Lock-free Execution), which is among the first GPU-native graph ANN indexes to support streaming insertion and deletion without a global rebuild.
  The core of ETALE is a lock-free copy-on-write slab graph whose deletion state and adjacency share a single atomically published word, which yields a provable deletion-monotonicity invariant together with a bounded reclaim of GPU memory that is sublinear in accumulated deletions.
  We have implemented ETALE in CUDA and evaluate it on five diverse multimodal datasets against four state-of-the-art indexes: Tagore (GPU static, SIGMOD'26), DIGRA (CPU dynamic, SIGMOD'26), CAGRA (GPU static, ICDE'24), and HNSW (CPU dynamic, TPAMI'20).
  Under continuous churn, ETALE maintains the index in hundreds of milliseconds per round at recall above $0.95$, which is $4.8$--$8.8\times$ faster than CAGRA's per-window rebuild, $1.8$--$2.5\times$ faster than the more recent Tagore, and $3.3$--$147\times$ faster than the CPU dynamic indexes.
  In addition, the memory footprint of ETALE stays bounded where tombstone-only systems grow indefinitely.
\end{abstract}

\maketitle

\section{Introduction}
\label{sec:intro}

Approximate nearest neighbor search (ANNS) over embedding vectors is a core
primitive behind many application domains, such as recommendation
systems~\cite{rying_kdd18} and retrieval-augmented generation
(RAG)~\cite{plewi_neurips20}. Across these settings, graph-based indexes such
as HNSW~\cite{hnsw}, NSG~\cite{cfu_vldb19}, and DiskANN~\cite{diskann} deliver
the desired accuracy-latency trade-off and have become a popular choice in
production vector stores~\cite{jwang_sigmod21}. The embedding models that
produce these vectors already run on GPUs, which gives a strong incentive to
keep the index on the GPU as well. A GPU-resident index sits next to the model
that generates the vectors, which removes the host transfer that a CPU index
pays on every batch and places graph search and construction on the same
parallel hardware.

However, keeping a graph index on the GPU and ensuring it up to date are in tension. 
Because a live corpus adds and removes documents continuously, an index that serves it
must absorb that stream of insertions and deletions or fall out of date. The
dominant GPU graph indexes cannot: they are built for static collections.
CAGRA~\cite{cagra}, the state-of-the-art GPU index, fixes every node's neighbor
list in a single batch-oriented construction pass, so it has no incremental
update path and absorbs a change only by rebuilding the affected window from
scratch. The indexes that do update incrementally sit on the other hardware:
FreshDiskANN~\cite{asing_arxiv21}, SPFresh~\cite{yxu_sosp23}, and
DIGRA~\cite{digra} show that streaming insertion and deletion are practical,
yet each maintains its graph on the RAM or SSD with CPUs and leaves the GPU's
parallelism unused. The cost of resolving the tension by rebuilding is large.
On SIFT under a $1\%$ churn workload, Figure~\ref{fig:teaser} shows CAGRA
spending about $1.8$\,s per round on the rebuild, whereas the proposed ETALE
index that maintains the graph incrementally finishes the same change in a
fraction of that time at matched recall. This gap follows from how each index
is built rather than from how it is tuned.

\begin{figure}[t]
  \centering
  \includegraphics[width=\linewidth]{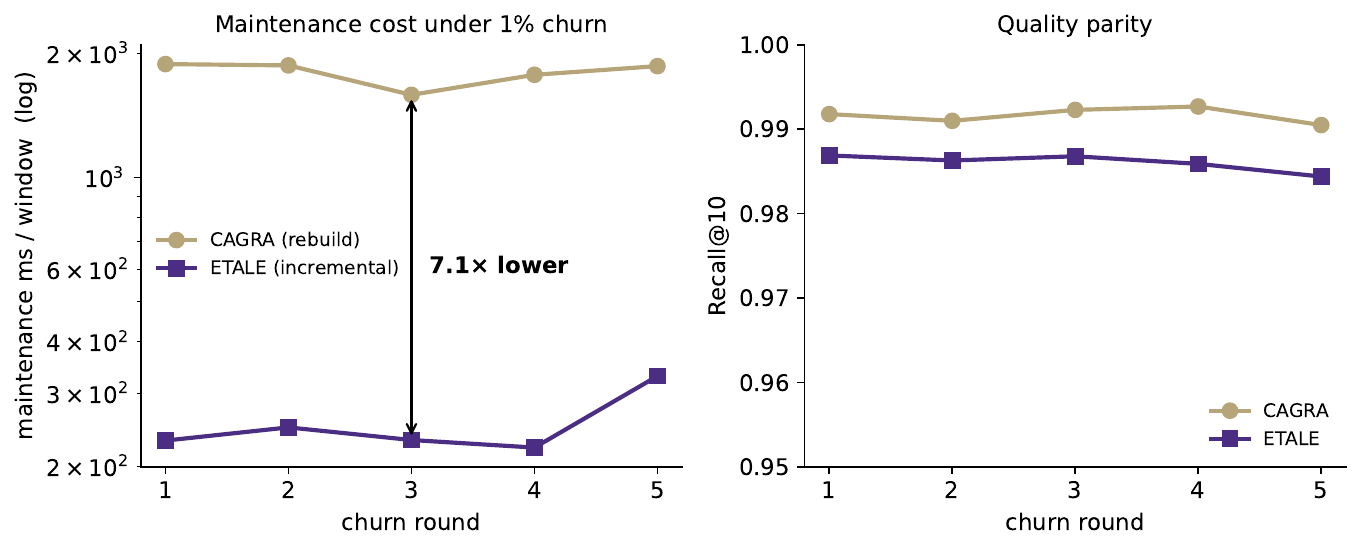}
  \caption{Maintenance under $1\%$ churn on SIFT. Left: ETALE maintains the
    index incrementally at about $250$\,ms per round, $7.1\times$ cheaper than
    CAGRA's per-round rebuild. Right: the saving holds at quality parity, with
    Recall@10 of the two systems coinciding within a few thousandths across rounds.}
  \label{fig:teaser}
\end{figure}

The technical challenge to realize GPU dynamism for ANNS lies in concurrent structural
mutation of a graph without locks. 
A node's update is not a single-word change
since it rewrites a whole adjacency list that no single atomic compare-and-swap
on a mutable array can publish, while fine-grained locking serializes exactly
the thousands-of-threads parallelism the GPU exists to provide. Recent GPU
efforts that add streaming insertion~\cite{ysun_cikm24} do so for an
inverted-file index, where an insertion appends a vector to a cluster list
rather than rewriting a node's adjacency, and they support insertion alone
without a deletion path. What has been missing is an index that mutates a
proximity \emph{graph} online on the GPU, supports deletion as well as
insertion, and whose deletion is provably monotone and whose footprint is
provably bounded.

This paper presents ETALE (Evolving Topology with Accelerated Lock-free
Execution), a graph ANN index that lives on the GPU and updates itself as data
arrives or leaves. It edits the graph in place on the GPU while queries continue to run,
which brings the cost of a stream of insertions and deletions to a fraction of
a rebuild. Its memory stays bounded as deletions
accumulate, where a conventionally tombstone-only index would grow indefinitely. A deleted
item, moreover, is guaranteed never to reappear in a result, which is a
deletion-monotonicity property that ETALE establishes by proof rather than by
heuristic. To our knowledge, ETALE is among the first GPU-native graph indexes
to support online streaming insertion and deletion and to back its deletion
with a monotonicity guarantee.

The core design of ETALE rests on a lock-free copy-on-write slab structure.
Each node's deletion state, degree, and adjacency pointer are packed into a
single $64$-bit word. The node's neighbors live in a separate immutable slab
that this word names. An update allocates a fresh slab, builds the new neighbor
set into it, and publishes the result with one atomic compare-and-swap on that
word. Publication of a full adjacency rewrite thus reduces to a single
compare-and-swap, which makes concurrent online insertion and deletion
lock-free without serializing the thousands of threads a GPU runs. Because the
deletion bit occupies the very word every update swaps, no later publication
can undo a deletion. Removal is therefore monotone by construction. On this
structure, we demonstrate a deletion-monotonicity invariant together with a
bounded reclaim of GPU memory. In addition, we confirm empirically that the
index is both fast and memory-stable.

In summary, this paper makes the following three contributions.
\begin{itemize}
  \item We design a new graph data structure for streaming ANNS on GPUs, namely ETALE.
        ETALE is a lock-free copy-on-write slab graph with a packed $64$-bit node
        reference, in which a full adjacency rewrite is published by a single
        compare-and-swap and the deletion bit is inlined into that same atomic word.
        The structure enables fine-grained online insert and delete rather than
        phase-structured batch updates. (Section~\ref{sec:structure})

  \item We establish the correctness and cost properties of ETALE. We prove a
        deletion-monotonicity invariant, reclaim safety, and a footprint bounded by the
        live set under periodic reclaim; and we analyze the reclaim cost, which is
        sublinear in accumulated tombstones under a uniform-placement model and is
        ceilinged by the cost of rebuilding the live set. (Section~\ref{sec:theory})

  \item We extensively evaluate ETALE with a full CUDA implementation and
        experimentation on five datasets against four state-of-the-art baselines, such
        as CAGRA (GPU static), Tagore (GPU static), DIGRA (CPU dynamic), and HNSW (CPU
        dynamic) under streaming churn. ETALE lowers maintenance cost by
        $4.8$--$8.8\times$ over CAGRA and $3.3$--$147\times$ over the CPU baselines at
        matched recall while keeping memory footprint bounded. (Section~\ref{sec:eval})
\end{itemize}


\section{Related Work}
\label{sec:related}


\subsection{Graph-Based ANN Search}
\label{sec:related-graph}

Graph-based methods give the desired accuracy-latency trade-off among ANN
indexes and underlie a lot of production vector search. A proximity graph connects
each vector to a set of near neighbors. A query walks the graph greedily from
an entry point toward its target, while a bounded beam of candidates keeps the
search from stalling in a local basin. HNSW~\cite{hnsw} layers navigable
small-world graphs~\cite{ymalk_is14} into a hierarchy that shortens the walk,
while NSG~\cite{cfu_vldb19} and Vamana within DiskANN~\cite{diskann} build a
single flat graph whose edges are chosen to stay navigable at a bounded degree.
FANNG~\cite{bharw_cvpr16} and clustering-based constructions~\cite{jmuno_pr19}
pursue the same trade-off through different edge-selection rules. ETALE adopts
this flat-graph design and the slack-pruned edge selection of Vamana, however re-crafts the slab-based graph structure and associated algorithms to better support GPU execution.

The quality of a graph index rests on how its edges are built and pruned.
NN-Descent~\cite{wdong_www11} and EFANNA~\cite{cfu_arxiv16} bootstrap an
approximate $k$-nearest-neighbor (kNN) graph efficiently, and Relative
NN-Descent~\cite{nono_mm23} and ParlayANN~\cite{mmano_ppopp24} parallelize the
construction deterministically. Recent work continues to refine the
construction and pruning that set a graph's reachable
recall~\cite{syang_vldb25, xzhao_vldb23, rqiu_sigmod25}, and a broad survey
compares the resulting designs~\cite{mwang_vldb21}. 
ETALE reuses the RobustPrune rule of DiskANN to bound out-degree at the slab capacity; however, the algorithms are adapted for GPU hardware and characteristics.

Graph methods displaced earlier families that ETALE does not build on but
shares a problem with. Locality-sensitive hashing (LSH) trades exactness for
sublinear query time~\cite{pindy_stoc98, mdata_socg04, qlv_vldb07,
  aando_neurips15}, and space-partitioning trees index low-dimensional data
well~\cite{jbent_cacm75, mmuja_tpami14}. Quantization compresses the vectors
themselves rather than the graph over them~\cite{sift, tge_tpami13,
  jgao_sigmod24, rguo_icml20, jmart_eccv18}, and libraries such as
Faiss~\cite{mdouz_arxiv24} combine these techniques in practice.

\subsection{Dynamic Graph ANN Indexes on CPUs}
\label{sec:related-dynamic}

A separate line of work maintains graph indexes incrementally, though on the
CPU or SSD rather than the GPU. FreshDiskANN~\cite{asing_arxiv21} and
SPFresh~\cite{yxu_sosp23} update a disk-resident graph in place as vectors
arrive and depart, DIGRA~\cite{digra} maintains a dynamic graph under
range-filtered queries, and a recent set of systems extends in-place streaming
maintenance to new settings~\cite{hxu_arxiv25, ypan_bigdata23}. These indexes
establish that incremental maintenance is feasible on host and disk hardware,
but none of them runs the maintenance on the GPU.

The hard part of maintenance is keeping the graph navigable as deletions remove
the nodes that paths route through. A lazy scheme tombstones a deleted node and
leaves its edges in place, which keeps deletion inexpensive but causes stale edges
to accumulate and memory to grow indefinitely. A repair-based scheme instead
rebuilds local connectivity around a deleted node. Several systems develop this
route through monotonic search-path repair~\cite{dliu_vldb25}, topology-aware
localized updates~\cite{syu_vldb25}, adaptive awareness of the changing
graph~\cite{jruan_kdd25}, and robustness to gradual content
drift~\cite{dbara_iccv23}. ETALE takes the repair-based route and additionally
reclaims the space that tombstones hold, which the CPU systems above leave to a
separate compaction pass.

The filtered queries that DIGRA targets sit at the boundary of this line.
Filtered-DiskANN~\cite{sgoll_www23} attaches attribute predicates to graph
edges, while ACORN~\cite{lpate_sigmod24}, VBASE~\cite{qzhan_osdi23}, and
AnalyticDB-V~\cite{cwei_vldb20} integrate vector search with structured or
relational predicates, and Starling~\cite{mwang_sigmod24} optimizes the disk
layout such hybrid queries read. ETALE does not support attribute filtering,
which scopes its comparison with DIGRA to the unfiltered maintenance both
systems perform.

\subsection{GPU ANN Indexes}
\label{sec:related-gpu}

Most GPU graph indexes build the graph once and query it on the GPU without an
incremental update path. CAGRA~\cite{cagra} constructs a high-quality graph
with a batch pipeline and rebuilds to absorb any change.
SONG~\cite{wzhao_icde20} and the GPU proximity-graph index of Yu et
al.~\cite{yyu_icde22} accelerate static graph search and construction,
BANG~\cite{kvenk_tbd25} answers billion-scale queries from a single GPU, and
further work scales static GPU indexing across larger inputs and
machines~\cite{hwang_cikm21, yzhu_sigmod24}.
Tagore~\cite{zli_sigmod25} belongs to this static line and accelerates the
construction of the same graph class, reducing build time without adding an
incremental insert or delete path. None of these indexes maintains the graph
under a stream of deletions, each absorbing an update by rebuilding.

A smaller line of work adds streaming updates on the GPU, and it is the closest
prior work to ETALE. RTAMS-GANNS~\cite{ysun_cikm24} maintains a GPU
inverted-file (IVF) index under concurrent insertion and search, appending new
vectors to a cluster's posting list and overlapping the insertion kernel with
the search kernel. It supports insertion only, with no deletion path and no
recall report. SIVF~\cite{dzhao_hpdc26} adds in-place deletion to a GPU IVF
index through conflict-free slab allocation, evicting expired vectors without
the host transfers that stall a static layout. Both index by inverted file
rather than by proximity graph, where an update edits a flat posting list and
never rewrites a node's adjacency. The concurrent adjacency-rewrite problem of
Section~\ref{sec:barrier} thus stays outside their scope. That line of work
sets graph-based methods aside as harder to maintain online. ETALE enters
exactly this region. A
full adjacency rewrite is published by a single compare-and-swap per node,
deletion proceeds by connectivity repair, and the structure carries a proof that
deletion is monotone and that reclaim is safe and bounded. The lock-free
publication that makes this online mutation possible rests on the GPU's relaxed
memory model~\cite{jalgl_asplos15}, and the pool that recycles retired slabs
follows the parallel-primitive discipline of shared-memory
toolkits~\cite{gblel_spaa20}.

\section{ETALE Design}
\label{sec:structure}



\subsection{Problem Formulation}
\label{sec:barrier}

The core problem ETALE aims to address is to allow many GPU threads to read and
update one graph of embedding vectors (i.e., nodes) at the same time, without
mutual exclusion. The threads fall into two classes: (i) readers that traverse
adjacency during search, and (ii) writers that rewrite adjacency during
insertion and deletion. Both classes run concurrently against the same nodes,
and neither should be made to wait for the other.

Correctness requires that each update to a node's adjacency take effect as a
single transaction, which is observed atomically by every concurrent reader. Suppose instead that a reader can observe a node while a writer is partway through
replacing its neighbor list. The reader then follows a mixture of old and new
neighbors that the graph never actually contained, in which case its search is led astray. 
A reader must therefore see a node's adjacency either entirely before
an update or entirely after it, i.e., in a transaction.

Although concurrency control is a well-studied topic in the database literature, committing a proximity-graph neighbor update under concurrent ANN search is not served well by standard techniques on a GPU. 
Lock-based (passive) concurrency control serializes conflicting
updates through a lock, but GPU threads execute in warps of 32 lanes in
lockstep.
Therefore, a lane that blocks on a lock holds up the other 31 lanes in its
warp and a high-degree node updated by many threads at once forces those
updates to serialize behind one lock. 
Lock-free (optimistic) concurrency control commits an
update with an atomic compare-and-swap (CAS), but a CAS commits a change to one
location only, whereas a neighbor-list rewrite must replace many entries at
once. 
The problem is thus to make a full neighbor-list rewrite to take effect as a
single atomic transaction, while preserving the thread parallelism a GPU index
depends on.

\subsection{Index Structure}
\label{sec:defn}

ETALE represents a directed proximity graph over a fixed identifier space. We
write $R$ for the slab capacity, i.e., the maximum out-degree of any node, $C$ for
the size of the identifier space, and $m$ for the vector dimension. The
identifier space is $\mathcal{U} = \{0, 1, \dots, C-1\}$, the slab-identifier
space is $\mathcal{S} = \{0, 1, \dots, P-1\}$ where $P$ is the number of slabs
in the pool, and a deletion flag takes a boolean value in $\{0, 1\}$.
We start by defining the ETALE graph formally, as follows.

\begin{figure*}[t]
  \centering
  \includegraphics[width=\textwidth]{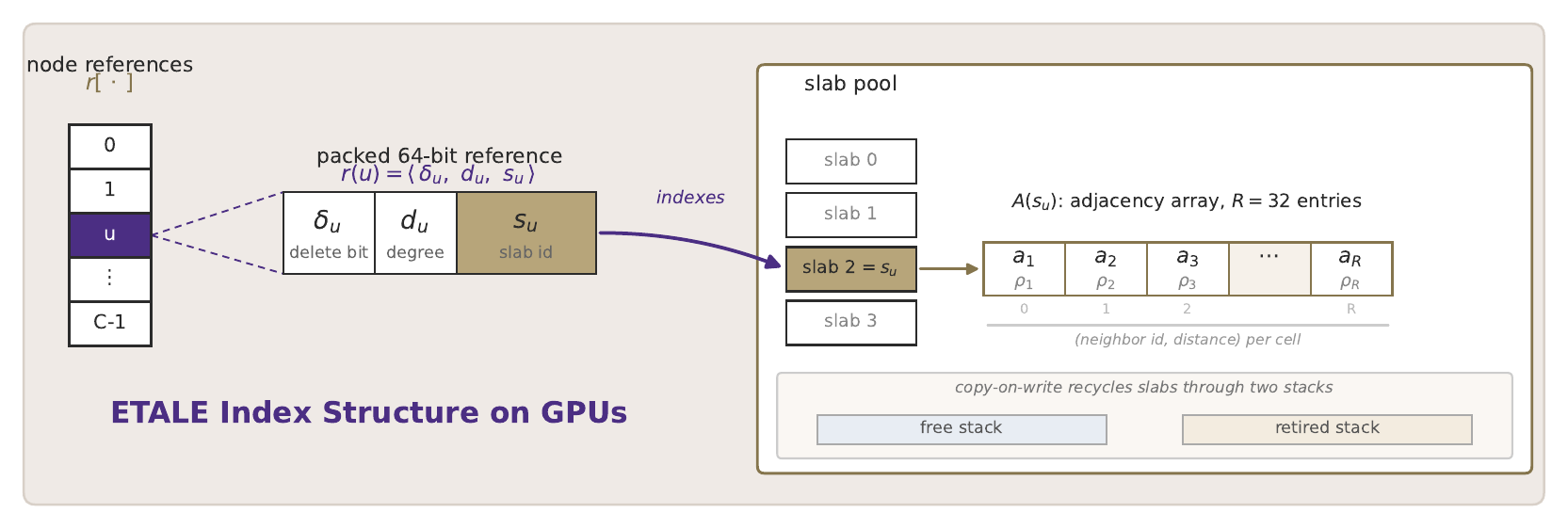}
  \caption{The ETALE index structure. Each identifier $u$ indexes a packed
    $64$-bit reference $r(u)=\langle\delta_u, d_u, s_u\rangle$, whose slab field
    $s_u$ names a slab in the pool. The slab holds $u$'s adjacency $A(s_u)$ as an
    array of up to $R$ entries, each a neighbor identifier paired with its
    distance. Copy-on-write publication draws fresh slabs from the free stack and
    returns retired ones through the retired stack.}
  \label{fig:structure}
\end{figure*}

\begin{definition}[ETALE graph]
  \label{def:graph}
  An ETALE graph is a tuple $G = (\mathcal{V}, r, A, x)$, where each component is defined as follows.
  \begin{itemize}
    \item $\mathcal{V} \subseteq \mathcal{U}$ is the set of live nodes.
    \item $r$ is the reference map, which assigns each identifier $u \in \mathcal{U}$ a node reference $r(u) = \langle \delta_u, d_u, s_u \rangle$,
          a triple whose components are a deletion flag $\delta_u \in \{0, 1\}$,
          an out-degree $d_u \in \{0, 1, \dots, R\}$,
          and a slab identifier $s_u \in \mathcal{S}$.
    \item $A$ is the slab map assigning each slab identifier $s \in \mathcal{S}$ an ordered sequence of $R$ entries,
          $A(s) = \big( (a_1, \rho_1), \dots, (a_R, \rho_R) \big)$,
          where each entry $(a_i, \rho_i)$ is a neighbor identifier $a_i \in \mathcal{U}$ paired with its distance $\rho_i \in \mathbb{R}$ to the node that owns the slab.
    \item $x$ is the vector map, which assigns each identifier $u \in \mathcal{U}$ its vector $x(u) \in \mathbb{R}^{m}$.
  \end{itemize}
  The out-neighbors of a live node $u$ are the neighbor identifiers in the first $d_u$ entries of its slab,
  $N(u) = (a_1, \dots, a_{d_u})$ where $(a_i, \rho_i) = A(s_u)[i]$.
  A node $u$ is a tombstone when $\delta_u = 1$.
  A tombstone is excluded from query results, while its slab is retained so that traversal may still relay through it.
  However, a tombstone's slab is eventually reclaimed and returned to the free pool when no live nodes reference it, at which point the tombstone is removed from the graph entirely.
\end{definition}

Figure~\ref{fig:structure} shows the three maps of Definition~\ref{def:graph}
laid out as ETALE holds them on the GPU. The representation reflects three
choices made for concurrent mutation. The first is to use a node's identifier
directly as its array index. A search that expands a neighbor reads that
neighbor's reference and vector by indexing on its identifier alone with no
lookup table on the path between an identifier and its data. 

The second choice is to place a node's deletion flag, degree, and slab
identifier together in one reference that the hardware accesses atomically. A
search reads all three in a single load and never sees them disagree. An update
commits all three at once. The deletion flag is therefore set, cleared, or
carried forward by the same commit that changes the adjacency rather than by a
second write to a separate structure.

The third choice is to hold a node's neighbors in a slab named by the
reference rather than inline in the node itself. An update can then prepare a
node's new neighbor list in a separate slab and substitute it by changing one
reference, which turns a change to a whole neighbor list into a change to a
single value. A search in progress reads either the old slab or the new one,
but never a list caught mid-edit. The slab holds up to $R$ neighbors, with $R$
fixed at the warp width of 32. A single warp reads a node's neighbors in one
coalesced step.

\subsection{Copy-on-write Publication}
\label{sec:cow}
Every ETALE update reaches the graph through one mechanism, the copy-on-write (COW) publication of Algorithm~\ref{alg:cow}.
A node's neighbors are changed not by editing its slab in place, but by building a fresh slab and redirecting the node's reference to it in a single atomic step.
A writer first snapshots the node's current reference,
then allocates an unused slab and writes the new neighbor sequence into the slab,
at which point the slab cannot be acquired by any other thread.
The writer issues a memory fence to order the slab's contents before any reference that points to them.
After that, the writer attempts a compare-and-swap (CAS) that moves the node's reference from the snapshot to a value identifying the fresh slab.
A successful swap then retires the old slab.
A failed swap implies that another writer published the node first, and the writer gives up this slab and retries from a new snapshot.
The reference carries forward the deletion flag of the snapshot; therefore, publication never revives a tombstone.

\begin{algorithm}[t]
  \caption{$\textsc{Publish}(u, N)$: install neighbor sequence $N$ for node $u$}
  \label{alg:cow}
  \SetKwInOut{Input}{Input}
  \SetKw{KwAnd}{and}
  \Input{node $u$, new neighbor sequence $N$ with $|N| \le R$}
  \Repeat{success \KwAnd attempt budget not exhausted}{
    $\langle \delta, d, s\rangle \gets r(u)$ \tcp*{atomic snapshot}
    $s' \gets \textsc{AllocSlab}()$\;
    $A(s') \gets N$ \tcp*{write fresh, still private}
    $\textsc{Fence}()$ \tcp*{content visible before reference}
    $r' \gets \langle \delta, |N|, s'\rangle$ \tcp*{carry the deletion flag}
    \uIf{$\textsc{CAS}\big(r(u),\, \langle\delta,d,s\rangle,\, r'\big)$}{
      $\textsc{RetireSlab}(s)$ \tcp*{won: free the old slab}
      success $\gets$ \textbf{true}\;
    }
    \Else{
      $\textsc{RetireSlab}(s')$ \tcp*{lost: discard and retry}
    }
  }
\end{algorithm}

This mechanism makes readers wait-free and writers lock-free, a property that
Section~\ref{sec:eval-concurrent} will demonstrate, i.e., in scenarios where queries run during concurrent updates. 
A reader's single atomic load yields a slab
identifier that names an immutable sequence, so the neighbor list it reads
cannot change beneath it. 
As a result, no torn read is possible and a writer always makes
progress. 
A failed swap witnesses another writer's success, which implies that the
system as a whole advances at every step.

The slabs come from a preallocated pool with two index stacks, a \emph{free stack} of
slab identifiers available for allocation and a \emph{retired stack} of slab
identifiers awaiting reclamation. $\textsc{AllocSlab}$ pops one identifier from
the free stack and $\textsc{RetireSlab}$ pushes one onto the retired stack.
Within a single kernel, the free stack is only popped and the retired stack is
only pushed. Neither stack is subject to a concurrent push and pop, and a
single atomic counter per stack serializes the threads that touch it. A retired
slab is not returned to the free stack during the kernel that retires it. The
host moves the retired stack onto the free stack at the boundary between
kernels, after the threads of the retiring kernel terminate and any
reader of a retired slab finishes. This separation of retirement from reuse
reclaims slabs without per-slab reference counting, which is an epoch discipline whose safety will be established in Section~\ref{sec:theory}.

\paragraph{Complexity Analysis}
A single $\textsc{Publish}$ runs in $O(R)$, linear in the slab capacity $R$
and constant in everything else. The snapshot, the allocation, the fence, the
CAS, and the retirement are each constant-time. The one $O(R)$ step is writing
the new neighbor sequence into the slab, which a warp of $R$ lanes performed in a
single coalesced pass. An uncontended publication therefore completes in a
constant number of parallel steps over that one slab. Under contention a
publication may retry, each retry rebuilding one slab at $O(R)$. A retry occurs
only after a competing publication on the same node has just committed. Each failed attempt then maps to a distinct successful one elsewhere. The total
number of committed publications is unchanged by retries. The amortized cost of
installing a node's neighbor list thus stays at $O(R)$.

\subsection{Update Protocols}
\label{sec:ops}

ETALE supports three updates, insertion, deletion, and reclamation.
Each computes a node's new neighbor sequence by combination of beam search, RobustPrune, and local connectivity repair, and finally commits it through $\textsc{Publish}$.
The three protocols differ only in which nodes they republish and in how they compute the new sequences.

\subsubsection{Insertion}
\label{sec:insert}

Insertion adds a node and connects it to the graph in both directions. It uses
two subroutines. $\textsc{BeamSearch}(v)$ takes a query vector $v$ and returns
the set of candidate nodes visited by a greedy search from a fixed set of entry
points. $\textsc{RobustPrune}_\alpha(\mathcal{C}, R)$ takes a candidate set
$\mathcal{C}$ and a degree bound $R$, and returns at most $R$ neighbors
selected from $\mathcal{C}$ under the slack-$\alpha$ pruning rule
of~\cite{diskann}, which keeps a near candidate only when no already-selected
neighbor dominates it. Algorithm~\ref{alg:insert} states the procedure.

The forward edges are selected once. A beam search from the new node gathers
candidates and RobustPrune reduces them to at most $R$ forward neighbors.
These become the new node's own slab. This first publication is uncontended since no other thread can yet reach the node.

The reverse edges then require a publication per forward neighbor, as an
edge from the new node alone does not make the new node reachable. For each
forward neighbor $w$, the new node is offered as a neighbor of $w$. A neighbor
$w$ below capacity gains the edge directly while a neighbor at capacity is
re-pruned over the union of its current neighbors and the new node; 
As a result, its degree never exceeds $R$. 
Nodes inserted in the same batch are absent from one another's beam searches. 
A refine pass runs afterward and re-offers each batch member's edges, which
restores the visibility the batch missed. 
A node whose neighbor sequence is unchanged by this pass publishes a no-op.

\paragraph{Complexity Analysis}
An insertion performs one beam search and at most $R+1$ publications, the node's
own slab and one reverse edge per forward neighbor. The beam search dominates
the cost. Each publication is the $O(R)$ operation analyzed in
Section~\ref{sec:cow}. The reverse publications are independent across forward
neighbors and proceed concurrently, each contending only on its own target.

\begin{algorithm}[t]
  \caption{$\textsc{Insert}(v)$: insert vector $v$ as a new node}
  \label{alg:insert}
  \SetKwInOut{Input}{Input}
  \Input{vector $v$ assigned a fresh identifier}
  $\mathcal{C} \gets \textsc{BeamSearch}(v)$ \tcp*{candidate neighbors}
  $F \gets \textsc{RobustPrune}_\alpha(\mathcal{C}, R)$ \tcp*{forward neighbors}
  $\textsc{Publish}(v, F)$ \tcp*{node's own slab, uncontended}
  \ForEach{$w \in F$}{
    $\langle\delta,d,s\rangle \gets r(w)$\;
    \uIf{$d < R$}{
      $N \gets A(s) \cup (v,\, \|x_v - x_w\|)$\;
    }
    \uElse{
      $N \gets \textsc{RobustPrune}_\alpha\!\big(A(s) \cup \{(v,\cdot)\},\, R\big)$\;
    }
    $\textsc{Publish}(w, N)$ \tcp*{reverse edge, lock-free}
  }
\end{algorithm}

\subsubsection{Deletion with connectivity repair}
\label{sec:delete}

Deletion removes a node from the vector database while keeping the graph navigable
around it. It reuses $\textsc{RobustPrune}_\alpha$ from insertion and adds no
new subroutine. Algorithm~\ref{alg:delete} states the procedure.

The node is first tombstoned by a single compare-and-swap (CAS) that sets its deletion flag while preserving its slab and degree.
A tombstoned node is excluded from query results because search admits to its output only nodes whose sampled reference has flag $0$.
The slab itself is retained so that a search already in progress can still relay through the node to reach others.
This step changes only the flag and remains a direct CAS rather than a full $\textsc{Publish}$.

Tombstoning alone would strand every path that routed through the deleted node, which is why the algorithm repairs the local connectivity around it.
For each live neighbor $a$ of the deleted node, the deleted node's remaining neighbors are the natural replacements for the two-hop reach that $a$ loses, from which RobustPrune selects the edges worth adding.
A single publication then installs $a$'s new neighbor sequence, comprising its surviving neighbors together with the selected replacements and excluding the deleted node.
This one publication per live neighbor rebuilds the bypassing edges and purges the edge into the tombstone within the same atomic step.
The repair strategy follows the local-repair line of~\cite{dzhao_tpami09}, while its realization as lock-free online single-reference edits is specific to ETALE and makes the removal monotone by construction through the inlined deletion flag (Proposition~\ref{prop:mono}).
An edge into the tombstone from a node that was not its neighbor is left for lazy cleanup, removed when that source next publishes and drops the tombstoned target.

\paragraph{Complexity Analysis}
A single deletion performs one CAS for the tombstone and one publication for each of the up to $R$ live neighbors repaired.
Each repair runs a RobustPrune over at most $R$ candidates followed by one $O(R)$ publication, which amounts to $O(R^2)$ work per deletion.
The per-neighbor repairs remain independent and execute concurrently.

\begin{algorithm}[t]
  \caption{$\textsc{Delete}(u)$: tombstone $u$ and repair connectivity around $u$}
  \label{alg:delete}
  \SetKwInOut{Input}{Input}
  \Input{live node $u$}
  $\langle\delta,d,s\rangle \gets r(u)$\;
  $\textsc{CAS}\big(r(u),\, \langle 0,d,s\rangle,\, \langle 1,d,s\rangle\big)$ \tcp*{tombstone, keep slab}
  $B \gets N(u)$\;
  \ForEach{$a \in B$ \textbf{with} $\delta_a = 0$}{
    $K \gets \textsc{RobustPrune}_\alpha\big(\{\, b \in B : b \neq a \,\},\, R\big)$ \tcp*{fix $a$}
    $N \gets \textsc{RobustPrune}_\alpha\big(\,(A(s_a) \setminus \{u\}) \cup K,\, R\,\big)$\;
    $\textsc{Publish}(a, N)$ \tcp*{rebuild and purge in one step}
  }
\end{algorithm}

\subsubsection{Reclamation}
\label{sec:reclaim}

Reclamation bounds the footprint by freeing the slabs that tombstones hold as relays, which sustained churn would otherwise cause to accumulate without bound.
It adds no new subroutine and reuses $\textsc{Publish}$ and $\textsc{RetireSlab}$ from Section~\ref{sec:cow}.
Algorithm~\ref{alg:reclaim} states the procedure as a policy-triggered pass in two ordered phases.

The first phase, \emph{compaction}, removes every edge from a live node to a tombstone.
Each live node whose neighbors include a tombstone is republished with the tombstoned entries dropped from its sequence.
After compaction, no live node references any tombstone.

The second phase, \emph{reclaim}, frees the now-orphaned tombstones.
Each tombstone's reference is cleared by a CAS that retires its slab to the pool.
The order of the two phases is a safety contract.
A tombstone's slab may be freed only once no live node can still follow an edge into it, which is exactly the postcondition that compaction establishes.

\paragraph{Complexity Analysis}
A single reclamation republishes at most one slab per live node in compaction and retires one slab per tombstone in reclaim.
This gives $O(L + t)$ publications over a live set of size $L$ holding $t$ tombstones.
The number of live nodes that compaction actually republishes equals the number that have a tombstoned neighbor.
This count is smaller than $L$ and grows sublinearly in $t$, which is the property that Section~\ref{sec:theory} will demonstrate and that bounds the cost of deferring reclamation.

\begin{algorithm}[t]
  \caption{$\textsc{Reclaim}()$: bound the footprint after churn}
  \label{alg:reclaim}
  \tcc{Phase 1, Compaction}
  \ForEach{live node $u$ \textbf{in parallel}}{
    \uIf{$\exists\, b \in N(u) : \delta_b = 1$}{
      $N \gets \{\, (b,\rho) \in A(s_u) : \delta_b = 0 \,\}$\;
      $\textsc{Publish}(u, N)$\;
    }
  }
  \tcc{Phase 2, Reclaim}
  \ForEach{tombstone $u$ \textbf{in parallel}}{
    \uIf{$\textsc{CAS}\big(r(u),\, \langle 1,d_u,s_u\rangle,\, \langle 1,0,\bot\rangle\big)$}{
      $\textsc{RetireSlab}(s_u)$\;
    }
  }
\end{algorithm}

\section{Correctness and Properties}
\label{sec:theory}


\subsection{Execution Model}
\label{sec:model}

The guarantees in this section are stated over a model of concurrent execution
with two levels, namely \emph{launches} and \emph{executions}, defined as follows.

\begin{definition}[Launch and execution]
  \label{def:exec}
  A \emph{launch} is a finite set of threads that run concurrently, each thread being a finite sequence of \emph{atomic steps}, where an atomic step is one aligned single-word load, one aligned single-word store, or one single-word compare-and-swap (CAS).
  The (atomic) steps of the threads in a launch interleave in an arbitrary order.
  A launch \emph{terminates} once every one of its threads has terminated.
  An \emph{execution} is a finite sequence of launches $\ell_1, \ell_2, \dots, \ell_n$ in which $\ell_{i+1}$ begins only after $\ell_i$ has terminated.
\end{definition}

In the context of GPUs, the above definition corresponds to the following. A
launch is one kernel invocation, in which many GPU threads are the lanes of
that kernel's warps. The execution of consecutive launches, also known as
serialization, is the series of device synchronization between kernels. ETALE
issues the maintenance of a round and the reclamation that may follow it as
separate launches; by doing so, according to Definition~\ref{def:exec}, no
maintenance step and no reclamation step ever interleave.

The serialization of launches enables the slab pool of Section~\ref{sec:cow} to
recycle memory without per-slab reference counting. A \emph{retirement} is a
device step that pushes a freed slab onto the retired stack and never touches
the free stack. A retired slab can return to the free stack only through the
host-side recycle step that runs between launches. A slab retired during a
launch stays unavailable for reallocation until that launch has terminated. We
formally state this property as the following lemma.

\begin{lemma}[Epoch separation]
  \label{lem:epoch}
  A slab $\sigma$ retired by some thread during launch $\ell_i$ is not reallocated at any step of $\ell_i$.
  No thread that reads $\sigma$ during $\ell_i$ would read it again after a reallocation.
\end{lemma}

\begin{proof}
  A reallocation of $\sigma$ is an $\textsc{AllocSlab}$ step that pops $\sigma$ from the free stack, which can occur only while $\sigma$ is present in that stack.
  By Definition~\ref{def:exec}, the retirement of $\sigma$ and every read of $\sigma$ are steps of $\ell_i$ that complete before $\ell_i$ terminates.
  The retirement places $\sigma$ on the retired stack, from which $\sigma$ rejoins the free stack only at the host-side recycle step.
  This recycle runs by construction after $\ell_i$ terminates and before $\ell_{i+1}$ begins.
  The slab $\sigma$ is therefore absent from the free stack throughout $\ell_i$.
  No $\textsc{AllocSlab}$ step of $\ell_i$ can pop it. The earliest reallocation of $\sigma$ therefore lies in $\ell_{i+1}$ or later.
  Every read of $\sigma$ during $\ell_i$ precedes the termination of $\ell_i$, which precedes the recycle of $\sigma$ that must complete before any reallocation.
  The reads and the reallocation are thus separated by the boundary of $\ell_i$ and cannot overlap.
\end{proof}

\subsection{Deletion Monotonicity}
\label{sec:mono}

A dynamic index must not resurrect a deleted node even as concurrent insertions
and repairs republish references throughout the graph. ETALE guarantees this
without any coordination beyond the single-word compare-and-swap (CAS) because
the deletion flag resides in the same word that every publication swaps. We
state this formally in the following proposition, which we call \emph{deletion
monotonicity}.

\begin{proposition}[Deletion monotonicity]
  \label{prop:mono}
  If a node $u$ has $\delta_u = 1$ at some point in an execution, then $\delta_u = 1$ holds at every later point and $u$ is absent from every query result returned thereafter.
\end{proposition}

\begin{proof}
  By Definition~\ref{def:graph}, the flag $\delta_u$ is one field of the single word $r(u)$. By Definition~\ref{def:exec}, that word changes only through a successful CAS.
  It therefore suffices to show that every successful CAS on $r(u)$ after $\delta_u$ first equals $1$ writes a value whose flag is again $1$.
  Recall that only three operations modify $r(u)$, as follows.
  (i) In $\textsc{Publish}$ (Algorithm~\ref{alg:cow}), the installed word is $\langle \delta, |N|, s'\rangle$, where $\delta$ is the flag read in the snapshot that the CAS tests for equality.
  A successful CAS certifies that $r(u)$ did not change between the snapshot and the swap, so $\delta$ still equals the current flag, namely $1$.
  (ii) In $\textsc{Delete}$ (Algorithm~\ref{alg:delete}), the single CAS sets the flag to $1$.
  (iii) In $\textsc{Reclaim}$ (Algorithm~\ref{alg:reclaim}), the CAS replaces $\langle 1, d, s\rangle$ by $\langle 1, 0, \bot\rangle$ and leaves the flag at $1$.
  No operation produces a successful swap that clears the flag.
  That is, once $\delta_u$ reaches $1$, it remains $1$.

  It remains to show that $u$ is absent from query results. Recall that a search
  returns a node only when the reference it samples has flag $0$. This implies
  that a node with $\delta_u = 1$ is never returned, even though its slab is
  retained to relay traversal. Once $\delta_u$ reaches $1$, the node $u$ is
  absent from every query result returned thereafter.
\end{proof}

The monotonicity of deletion is a direct consequence of inlining the deletion
flag in the same word that every publication swaps. It is a property of a
single node's word that is established with no lock and no multi-word transaction. A
representation that placed the deletion flag in a separate array would instead
have to keep that array consistent with the adjacency under concurrent
publication. The inlined flag design of ETALE removes exactly that obligation.

\subsection{Reclaim Safety}
\label{sec:safety}

Tombstoning retains a deleted node's slab such that concurrent searches may
still relay through it. While reclamation eventually frees that slab, safety
requires that no slab is freed while a search may still follow an edge into it.
The guarantee rests on the two-phase order of $\textsc{Reclaim}$ together with
the epoch separation of Lemma~\ref{lem:epoch}. We formalize this property in
the following proposition, which we call \emph{no use-after-free}.

\begin{proposition}[No use-after-free]
  \label{prop:safety}
  No thread follows an edge into a slab that has been reclaimed and reallocated.
\end{proposition}

\begin{proof}
  Consider a tombstone $u$ whose slab is freed in the reclaim phase.
  A search reaches the slab of $u$ only by following an edge $w \to u$ from a node $w$ it is visiting.
  After the compaction phase, no live node holds an edge into a tombstone since compaction republishes every live node whose adjacency contains a tombstoned neighbor and drops those neighbors from the new sequence (Algorithm~\ref{alg:reclaim}).
  The only remaining references into the slab of $u$ are from other tombstones, which a search never admits as results and from which no live traversal originates.
  The reclaim phase retires the slab of $u$ at a step of the reclamation launch.
  By Lemma~\ref{lem:epoch}, that retired slab is not reallocated before the reclamation launch terminates, hence not before every thread of that launch has terminated.
  By Definition~\ref{def:exec}, the searches run in launches that synchronization separates from the reclamation launch.
  Any search that reaches the slab therefore terminates before the reclamation launch begins,
  whereas the slab is reallocated only in a launch after the reclamation.
  Therefore, no search can read the slab once it is reallocated.
\end{proof}

\subsection{Bounded VRAM Footprint}
\label{sec:footprint}

We now show that the resident GPU memory (VRAM) is bounded. A node is live when
$\delta = 0$. We write $L$ for the number of live nodes. Every node, live or
tombstoned, owns exactly one slab. A publication allocates one fresh slab and
retires the displaced one. 
Between reclamations, the resident slab count equals
$L$ plus the number of tombstones that have not yet been reclaimed.

\begin{proposition}[Bounded memory footprint]
  \label{prop:footprint}
  Immediately after a reclamation completes, the resident slab count equals the live-set size $L$.
\end{proposition}

\begin{proof}
  By Proposition~\ref{prop:safety}, the reclaim phase safely frees the slab of every tombstone.
  By construction it frees no live node's slab.
  After compaction, no live node references a tombstone, so every tombstone is orphaned and its slab is retired in the reclaim phase.
  The recycle step returns the retired slabs to the free pool.
  The surviving slabs are exactly those of the $L$ live nodes, one apiece.
  The resident count is therefore $L$, as claimed.
\end{proof}

This bound property implies that the resident memory follows a sawtooth profile
under continuous churn. Between reclamations, the slab count climbs as deletions
accumulate tombstones that are not yet freed. Each reclamation then frees those
tombstones at once and resets the count to its lower bound $L$.
Section~\ref{sec:eval-memory} will verify this pattern on real workloads.

\subsection{Sublinear Reclaim Cost}
\label{sec:cost}

The remaining question is the cost of each reset. 
Clearly, the cost of
$\textsc{Reclaim}$ is dominated by compaction, which republishes every live
node that has at least one edge into a tombstone. A live node with such an edge
is called \emph{dirty}. A single republication cleans all of a node's dead edges at once;
therefore, a dirty node costs one republication regardless of how many dead edges it carries. 
Reclaim cost thus scales with the number of dirty nodes, not with
the number of dead edges. In fact, the number of dirty nodes grows sublinearly
in the number of tombstones because a tombstone's stranded in-edges tend to
concentrate on already-dirty nodes rather than spreading uniformly over the
live set. We formalize this property in the following proposition, which we
call \emph{sublinear reclaim}.

\begin{proposition}[Sublinear reclaim]
  \label{prop:reclaim}
  $D(t)$ denotes the expected number of dirty nodes after $t$ tombstones accumulate over a live set of size $L$ with mean in-degree $\bar d$.
  Under a model in which each tombstone strands its $\bar d$ in-edges on live nodes chosen independently and uniformly at random, $D(t)$ grows as
  \[
    D(t) \;=\; L\!\left(1 - \Bigl(1 - \tfrac{1}{L}\Bigr)^{\bar d\, t}\right)
    \;\approx\; L\!\left(1 - e^{-t/\tau}\right), \qquad \tau = \frac{L}{\bar d}.
  \]
  The reclaim cost is proportional to $D(t)$ and saturates as
  $\mathrm{reclaim}(t) = M\!\left(1 - e^{-t/\tau}\right)$, where $M$ is the cost
  of republishing all $L$ live nodes.
\end{proposition}

\begin{proof}
  A live node remains clean after $t$ tombstones precisely when none of the $\bar d\, t$ stranded in-edges falls on it.
  Since the stated model lands those in-edges independently and uniformly over the $L$ live nodes, a fixed node avoids any one of them with probability $1 - 1/L$, hence avoids all $\bar d\, t$ of them with probability $(1 - 1/L)^{\bar d t}$.
  Summing the complementary probability over the $L$ nodes gives the expected dirty count $D(t) = L\bigl(1 - (1 - 1/L)^{\bar d t}\bigr)$.
  For large $L$, the limit $(1 - 1/L)^{L} \to e^{-1}$ turns the inner factor into $e^{-\bar d t / L}$, so $D(t)$ takes the exponential form with $\tau = L / \bar d$.
  Because compaction republishes each dirty node once, its cost is proportional to $D(t)$ and inherits the same saturating shape with the per-republication constant absorbed into $M$.
  As $t$ grows without bound, every live node eventually becomes dirty.
  The count $D(t)$ then approaches $L$ and the cost approaches its ceiling $M$,
  which is the cost of republishing the entire live set once.
\end{proof}

The saturating form has a practical consequence for a long-running index.
Because the cost levels off rather than growing without bound, no amount of
deferred deletion can make a single reclamation arbitrarily expensive. The next
corollary states this ceiling exactly.

\begin{corollary}[Bounded reclaim cost]
  \label{cor:reclaim-bound}
  For every $t$, a single reclamation costs $\mathrm{reclaim}(t) < M$, the cost of rebuilding the entire live set once.
\end{corollary}

\begin{proof}
  The factor $1 - e^{-t/\tau}$ lies strictly below $1$ for every finite $t$.
  The cost $\mathrm{reclaim}(t) = M(1 - e^{-t/\tau})$ is therefore strictly below $M$.
\end{proof}

The bound holds however long reclamation is deferred. Deferral is in fact
favorable because the cost grows ever more slowly as deletions accumulate. The
cost $\mathrm{reclaim}(t) = M(1 - e^{-t/\tau})$ is concave, since its second
derivative $-\tfrac{M}{\tau^2}e^{-t/\tau}$ is negative for every $t$. The
mechanism behind this concavity is that a tombstone arriving later tends to
strand its edges on nodes that are already dirty, hence already scheduled for
one republication. 
Reclaiming at any interval below saturation thus costs strictly less than the linear estimate that multiplies the per-deletion cost of the early rounds by the number of deletions.

\begin{remark}[The Saturating Form as an Empirical Regularity]
  \label{rem:tau}
  Section~\ref{sec:eval-reclaim} will fit this saturating form with $R^2 = 0.999$ and confirm its premise that dirty nodes saturate toward $L$ while dead edges grow linearly.
  The fitted time constant is smaller than the model value $\tau = L / \bar d$, because stranded in-edges concentrate on high-degree hubs rather than spreading uniformly, which saturates the dirty set faster than the uniform model assumes.
  We therefore read the saturating form as an empirical regularity whose mechanism the model explains, not as a quantitative predictor of $\tau$.
\end{remark}

\section{Evaluation}
\label{sec:eval}


\subsection{System Implementation}

ETALE is implemented in CUDA C++ for the NVIDIA Ampere architecture.
A node's neighbors occupy a slab of $R = 32$ entries, whose width matches the warp so that one warp scans a node's adjacency in a single coalesced step.
All updates go through the copy-on-write publication (Section~\ref{sec:cow}) that a single 64-bit compare-and-swap on the node reference commits.
The slabs are drawn from the two-stack pool that recycles memory at kernel boundaries.

As of June 2026, ETALE comprises about $7{,}200$ lines of code.
Of these, roughly $4{,}300$ lines of C++ and CUDA implement the index and its kernels, while the remaining $2{,}900$ lines of Python and shell run the experiments.
The source code and the scripts that reproduce every figure in this section are available at \url{https://github.com/hpdic/etale}.

\subsection{Experimental Setup}
\label{sec:eval-setup}

\paragraph{Baselines}
We compare ETALE against four graph indexes summarized in
Table~\ref{tab:baselines}. On the GPU axis, CAGRA~\cite{cagra} from the cuVS
library is a state-of-the-art static GPU graph index, and
Tagore~\cite{zli_sigmod25} is a more recent GPU library that accelerates the
construction of the same class of graphs. Neither has an incremental update
path, so each absorbs a round of churn by rebuilding the affected window from
scratch. On the CPU axis, HNSW~\cite{hnsw} from hnswlib and DIGRA~\cite{digra}
are dynamic graph indexes that support incremental insertion and deletion, but
run on the host and do not use the GPU. ETALE is the only index in the
comparison that maintains a graph incrementally on the GPU.

\begin{table*}[t]
  \centering
  \caption{Baselines along the axes relevant to dynamic GPU search.
  }
  \label{tab:baselines}
  \setlength{\tabcolsep}{5pt}
  \begin{tabular}{l c c c c c c c c c c}
    \toprule
    System                     & Hardware & Dynamic & Lock-free & Insert  & Delete    & Reclaim & Bounded & Filter & Concurrent & Provable \\
    \midrule
    CAGRA~\cite{cagra}         & GPU      & \xmark  & \xmark    & rebuild & rebuild   & --      & --      & \xmark & \xmark     & \xmark   \\
    Tagore~\cite{zli_sigmod25} & GPU      & \xmark  & \xmark    & rebuild & rebuild   & --      & --      & \xmark & \xmark     & \xmark   \\
    HNSW~\cite{hnsw}           & CPU      & \cmark  & \xmark    & \cmark  & tombstone & \xmark  & \xmark  & \xmark & \xmark     & \xmark   \\
    DIGRA~\cite{digra}         & CPU      & \cmark  & \xmark    & \cmark  & \cmark    & \cmark  & \cmark  & \cmark & \xmark     & \xmark   \\
    ETALE (ours)               & GPU      & \cmark  & \cmark    & \cmark  & repair    & \cmark  & \cmark  & \xmark & \cmark     & \cmark   \\
    \bottomrule
  \end{tabular}
\end{table*}

\paragraph{Metrics}
We report two metrics throughout the evaluation, maintenance cost and recall.
Maintenance cost is the wall-clock time to absorb one round of churn, measured
per round and averaged over rounds one through five with the initial build
excluded. 
For ETALE, this is the time of the incremental insert, delete, and
reclaim that the round triggers. 
For the static GPU baselines, it is the time of
the per-round rebuild. Recall is Recall@10, the fraction of the ten exact
nearest neighbors over the current live set that a query returns, 
which is averaged over
the held-out query set and measured after each round so that any drift from
accumulated updates is visible. We additionally track the resident slab or
element count to test the GPU memory footprint bound under sustained churn.

\paragraph{Datasets}
We evaluate on five public datasets that span image and audio modalities,
summarized in Table~\ref{tab:datasets}. Each dataset is indexed with the active
set listed in Table~\ref{tab:datasets} and queried with the dataset's standard
held-out query set. Recall@10 is measured against the exact nearest neighbors
over the current live set.
The churn workload makes the active working set, rather than the base pool, the
quantity that determines cost. Maintenance time, device memory, and recall all
follow the active set that resides on the GPU, while the base pool only supplies
fresh vectors for insertion and needs to be large enough that churn does not
exhaust it. We therefore fix the active set per dataset and draw insertions from
the remaining base. SIFT illustrates the relationship to billion-scale data. Our SIFT base is the
standard million-scale subset of the public SIFT1B set, with the active working
set drawn from it. A billion-scale base such as SIFT1B or
DEEP1B extends how long churn can run before reusing vectors without changing the
per-round maintenance load. We therefore use million-scale bases without loss of
representativeness for this workload.

\begin{table}[t]
  \centering
  \caption{Datasets used in the evaluation.
  }
  \label{tab:datasets}
  \begin{tabular}{llrrrr}
    \toprule
    Dataset                      & Modality & Dim. & Base size     & Active    & Churn   \\
    \midrule
    SIFT~\cite{sift}             & Image    & 128  & 1{,}000{,}000 & 200{,}000 & 2{,}000 \\
    DEEP~\cite{deep}             & Image    & 256  & 1{,}000{,}000 & 200{,}000 & 2{,}000 \\
    MSONG~\cite{msong}           & Audio    & 128  & 1{,}000{,}000 & 200{,}000 & 2{,}000 \\
    Notre~\cite{dpg_audio_notre} & Image    & 128  & 332{,}668     & 200{,}000 & 2{,}000 \\
    Audio~\cite{dpg_audio_notre} & Audio    & 192  & 54{,}387      & 38{,}948  & 2{,}000 \\
    \bottomrule
  \end{tabular}
\end{table}

\paragraph{Platform}
All experiments run on a server with two AMD EPYC 7763 64-core processors (256
hardware threads in total) and 503\,GB of RAM, hosting a single NVIDIA A100
80\,GB PCIe accelerator (driver 560.35.05). The GPU code is built with CUDA
12.6. The host code is built with GCC 13.3 under Ubuntu 24.04.4 (Linux kernel
6.8.0). We run CAGRA through the cuVS Python package and HNSW through the
hnswlib implementation bundled with DIGRA~\cite{digra}.

\paragraph{Workload}
We measure maintenance under a streaming churn workload. An index is first
built on the live set. Each subsequent round inserts and deletes an equal
number of vectors, set to one percent of the live set per round unless stated
otherwise. We record the per-round maintenance time and the post-round
Recall@10. The reported numbers are the mean over rounds one through five.
Round zero (warmup) is the initial build and is excluded from the maintenance means.

\subsection{End-to-end Comparison with GPU Indexes}
\label{sec:eval-gpu}

We first compare ETALE against the GPU graph indexes that share its hardware,
namely CAGRA~\cite{cagra} and Tagore~\cite{zli_sigmod25}. CAGRA is the
established state-of-the-art GPU graph index, and Tagore is a more recent GPU
library that accelerates the construction of the same class of graphs. Neither
has an incremental update path, so each churn round rebuilds the graph over the
live set. We report this comparison on SIFT, MSONG, Notre, and DEEP. Audio is
held out here because its active set holds only $38{,}948$ vectors, a scale at
which a full rebuild is already light enough that it no longer separates from
incremental maintenance for either GPU baseline. Figure~\ref{fig:maint-gpu}
reports per-round maintenance time across the four datasets.

\begin{figure}[t]
  \centering
  \includegraphics[width=\linewidth]{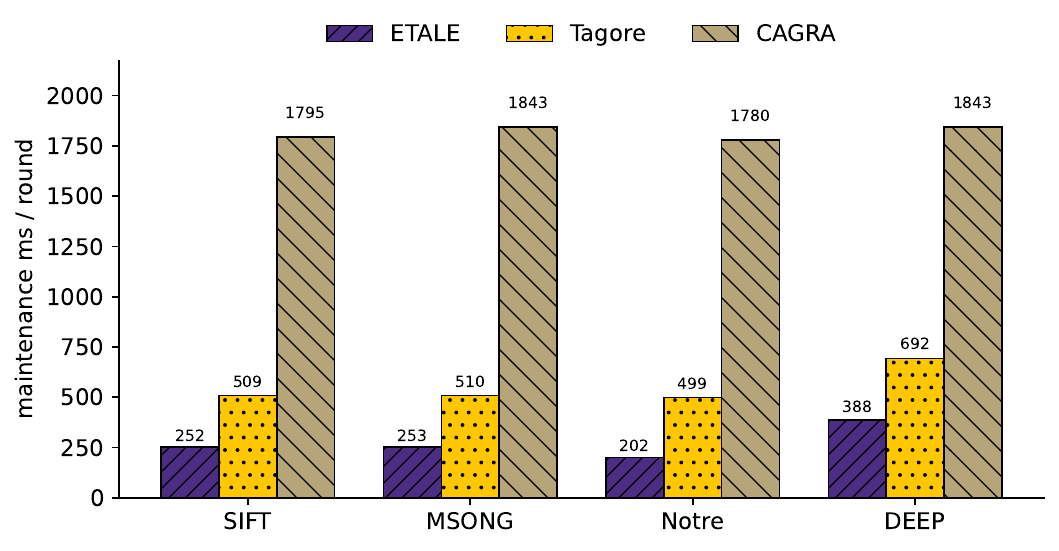}
  \caption{Per-round maintenance under $1\%$ churn on the GPU, ETALE
    against CAGRA and Tagore.
  }
  \label{fig:maint-gpu}
\end{figure}

ETALE maintains the index incrementally in $202$--$388$\,ms per round, whereas
CAGRA rebuilds in $1.78$--$1.84$\,s, a $4.8$--$8.8\times$ reduction. The gap is
largest on Notre ($8.8\times$) and smallest on the highest-dimensional set,
DEEP ($4.8\times$), where ETALE's own maintenance cost is highest. This speedup
holds at matched accuracy, with ETALE's Recall@10 ($0.953$--$0.988$) tracking
CAGRA's ($0.974$--$0.992$) within a few points across the four datasets.

Tagore rebuilds faster than CAGRA because its GPU-specific construction replaces the generic build pipeline with a two-phase descent for
k-NN graph initialization and parallel kernels for the pruning step;
yet ETALE still maintains the index $1.8$--$2.5\times$
faster than Tagore. The margin is smallest on DEEP ($1.8\times$) and largest on
Notre ($2.5\times$). Tagore's rebuild time is nearly flat at about $500$\,ms,
which is governed by the live-set size rather than the vector dimension. ETALE's
incremental cost instead falls with dimension and stays below Tagore's on every
dataset. ETALE thus outpaces not only CAGRA but also the faster rebuild that
Tagore represents, which confirms that the advantage comes from maintaining the
graph incrementally rather than from any one rebuild being slow.

\subsection{End-to-end Comparison with CPU Indexes}
\label{sec:eval-cpu}

The GPU comparison isolates the cost of rebuilding, but it leaves open whether
incremental maintenance is worth moving to GPUs at all. The natural point of
reference is the class of dynamic indexes that already support incremental
insertion and deletion on CPUs.
To answer this question, we compare ETALE against
HNSW~\cite{hnsw} and DIGRA~\cite{digra}, two graph indexes that maintain their
structure on the host under the same churn workload. 
A favorable comparison
here shows that the benefit of ETALE comes from GPU-resident maintenance, not
merely from avoiding rebuilds.

\begin{figure}[t]
  \centering
  \includegraphics[width=\linewidth]{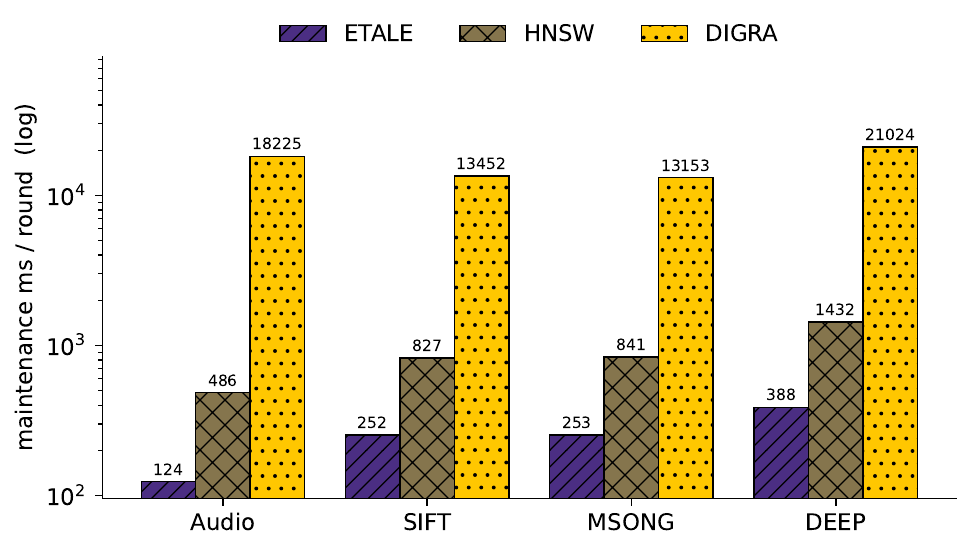}
  \caption{Per-round maintenance under 1\% churn, ETALE vs.\ CPU dynamic indexes
    HNSW and DIGRA, log scale.
  }
  \label{fig:maint-cpu}
\end{figure}

Figure~\ref{fig:maint-cpu} reports per-round maintenance on four datasets, with
every system held above $0.95$ Recall@10. ETALE maintains the index
$3.3$--$3.9\times$ faster than HNSW, the stronger of the two CPU baselines.
HNSW spends most of its maintenance time on a neighbor-heavy update path whose
throughput is bounded by CPU memory bandwidth, a limit that GPU-resident
adjacency avoids. The margin over DIGRA is far larger, at $52$--$147\times$.

DIGRA's higher maintenance cost can be explained by its original design
objective. It targets the more specific problem of range-filtered dynamic
search, where the periodic structural reorganization is the price of supporting
that filter. The comparison here isolates the narrower task of unfiltered
maintenance; therefore, the margin reflects this difference in scope rather
than a deficiency in DIGRA. This is also the reason why we omit Notre from the
DIGRA comparison: we could not find a DIGRA configuration that converges on it,
whereas the other systems all exceed $0.99$ Recall@10.

\subsection{Deletion: Memory and Navigability}
\label{sec:eval-memory}

A dynamic index must survive deletion on two fronts, the memory it holds as
nodes are removed and the navigability of the graph that the removals leave
behind. A graph index that deletes by tombstoning, as HNSW does with
\texttt{markDelete}, flips a bit and never frees the slot, leaving its resident
size to grow monotonically with cumulative insertions (i.e., unbounded memory growth). ETALE instead repairs
connectivity at deletion and reclaims retired slabs at a periodic interval,
which keeps its footprint bounded by the live-set size. We test the footprint
under balanced churn, then the recall rates under pure deletion.

\begin{figure}[t]
  \centering
  \includegraphics[width=\linewidth]{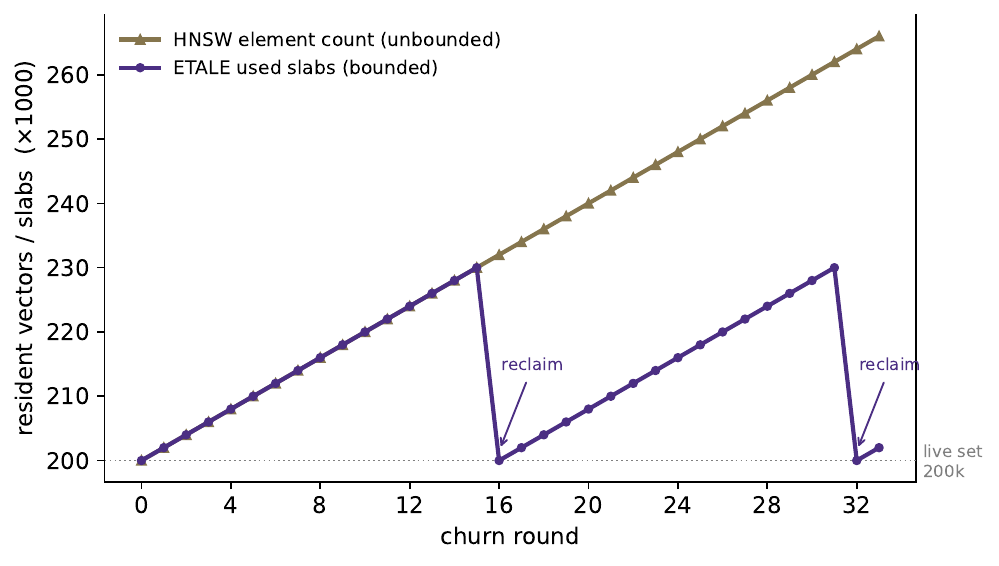}
  \caption{Resident vectors/slabs under churn on SIFT1M, both measured over $33$
    rounds at $1\%$ churn.}
  \label{fig:memory}
\end{figure}

Figure~\ref{fig:memory} contrasts the two footprints over a churn run. ETALE's
used-slab count climbs by the per-round insertion volume between reclaims and
drops back to the live-set size of $200$k whenever a reclaim fires, a sawtooth
whose ceiling the reclaim interval sets. The tombstone-only baseline never
returns to the live-set size and climbs without bound. By round $33$ HNSW holds
$266$k elements, a $33\%$ increase over the live set, against ETALE's exact
$200$k. This is the bounded-footprint property of Section~\ref{sec:theory} in
action, a consequence of the connectivity repair and reclamation ETALE performs.

\begin{figure}[t]
  \centering
  \includegraphics[width=\linewidth]{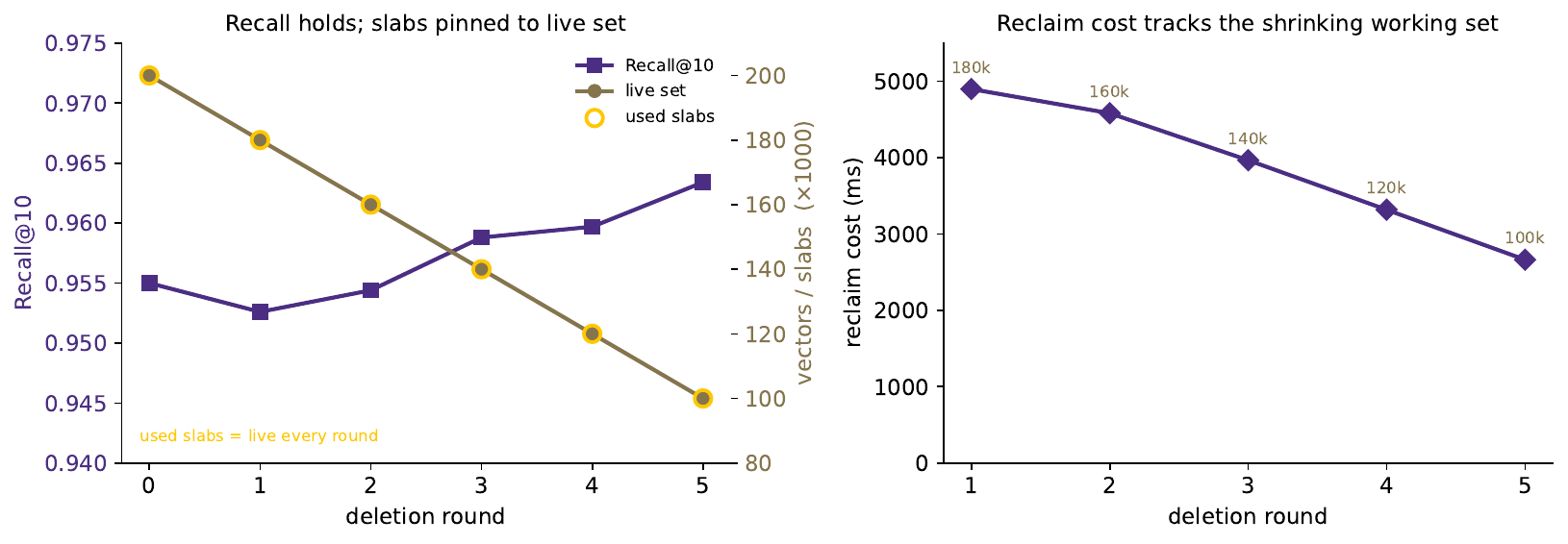}
  \caption{Pure-deletion stress on DEEP (delete $20$k/round, no insertions, to
    live = $100$k).}
  \label{fig:delete}
\end{figure}

Balanced churn refreshes the graph with new edges that can mask damage in the
delete path. A pure-deletion test removes that cover. We start from a
$200$k-vector index on DEEP and delete $20$k random live nodes per round with
connectivity repair and reclaim, no insertions, for five rounds until the live
set halves to $100$k. Recall then turns entirely on whether the delete path
preserves navigability.
Figure~\ref{fig:delete} reports recall, footprint,
and reclaim cost.

Recall survives the heavy deletion, rising slightly from $0.955$ at $200$k live
to $0.963$ at $100$k as a fixed query beam covers the smaller graph more
completely. The absence of decay with half the nodes gone is the point. The
deletion-monotonicity invariant of Proposition~\ref{prop:mono} underwrites it: a
deleted node is excluded from results but still relays traversal, its in-edges
repaired so a path through it is rebuilt around it rather than severed.

The footprint again holds to the live set, used slabs tracking it from $180$k
down to $100$k with no residual growth. Pure deletion is the strongest case of
the bound, the structure shrinking to exactly the working set. The reclaim cost
falls with that set, from $4.9$\,s at $180$k live to $2.7$\,s at $100$k, because
compaction repairs only the dirty live nodes and its cost follows the
working-set size rather than the cumulative deletions, matching the sublinear
characterization of Section~\ref{sec:eval-reclaim}.

\subsection{Query Latency under Concurrent Update}
\label{sec:eval-concurrent}

The results so far measure recall after each round completes, which leaves the
central concurrency claim untested at the system level. The wait-free reader
argument of Section~\ref{sec:cow} holds at the data-structure level, since a
reader's single atomic load names an immutable slab. This experiment tests
whether that property survives end to end, with queries in flight while the
graph mutates.

We run search and update on two CUDA streams that are never cross-synchronized.
The device is then free to overlap them. The query stream issues batches of
$1{,}000$ queries back to back and times each batch. The update stream
concurrently runs $200$ rounds of balanced delete and insert at $2{,}000$
operations each on DEEP. We measure the query-batch latency distribution in two
settings: quiescent with no update in flight, and concurrent with the update
stream running. Table~\ref{tab:concurrent} reports both. Recall is measured
against a fixed anchor set whose ground truth does not change under the churn,
so that a query overlapping an update is scored correct precisely when it
returns the true neighbors, independent of how far the update stream has
progressed.

\begin{table}[t]
  \centering
  \caption{Query latency on DEEP with and without a concurrent update stream, per batch of $1{,}000$ queries.}
  \label{tab:concurrent}
  \setlength{\tabcolsep}{8pt}
  \begin{tabular}{l c c c}
    \toprule
                           & p50        & p99        & Recall@10 \\
    \midrule
    Quiescent              & $15.9$\,ms & $16.6$\,ms & $0.954$   \\
    Concurrent with update & $16.6$\,ms & $16.8$\,ms & $0.954$   \\
    \bottomrule
  \end{tabular}
\end{table}

The query tail stays essentially flat under concurrent mutation.
The p99 batch latency rises from $16.6$\,ms quiescent to $16.8$\,ms with $200$ delete-and-insert rounds in flight, an inflation of $1.01\times$. The p50 rises by a comparable $0.7$\,ms.
Both shifts are small and point the same way, which is the signature of two kernels sharing the device rather than one blocking the other.
A reader that contended with writers for a lock would instead show a tail that grows with the update load.
The flat tail here reflects the wait-free path that a single atomic load gives every reader, where the residual cost is the device dividing its multiprocessors between the query stream and the update stream.

Recall holds at the quiescent level throughout the concurrent run.
Queries that overlap the update stream return Recall@10 of $0.954$ against the fixed anchor ground truth, which matches the $0.954$ measured with no update in flight.
A query that observed a torn adjacency, a half-installed slab, or a reclaimed slab would miss its true neighbors and depress this number.
Its stability confirms that the copy-on-write publication presents every concurrent reader a consistent graph.
The two measurements together move the concurrency property from the data-structure argument of Section~\ref{sec:cow} to an end-to-end system result.

\subsection{Sublinear Reclaim Cost}
\label{sec:eval-reclaim}

Section~\ref{sec:theory} claims that reclaim cost saturates as accumulated
tombstones grow. The cost is set by the number of dirty nodes rather than the
number of dead edges, where a dirty node is a live node with at least one dead
edge. Because dirty nodes saturate toward the live-set size while dead edges
grow linearly, the cost flattens even as tombstones keep accumulating. We
validate this by sweeping the reclaim interval and measuring a single reclaim
at each accumulated-tombstone level. 

Figure~\ref{fig:reclaim-fit} fits the
measured reclaim cost to the saturating form
\[
  \mathrm{reclaim}(t) = M \cdot \bigl(1 - e^{-t/\tau}\bigr),
\]
obtaining $M \approx 658$\,ms and $\tau \approx 1.6\times10^{4}$ tombstones at
$R^2 = 0.999$, where $t$ is the number of accumulated tombstones and $e$ is the
base of the natural logarithm. Figure~\ref{fig:reclaim-mech} shows the
mechanism behind the fit, with dirty nodes saturating toward the $200$k live
set while dead edges grow linearly.

\begin{figure}[t]
  \centering
  \begin{subfigure}{0.49\linewidth}
    \includegraphics[width=\linewidth]{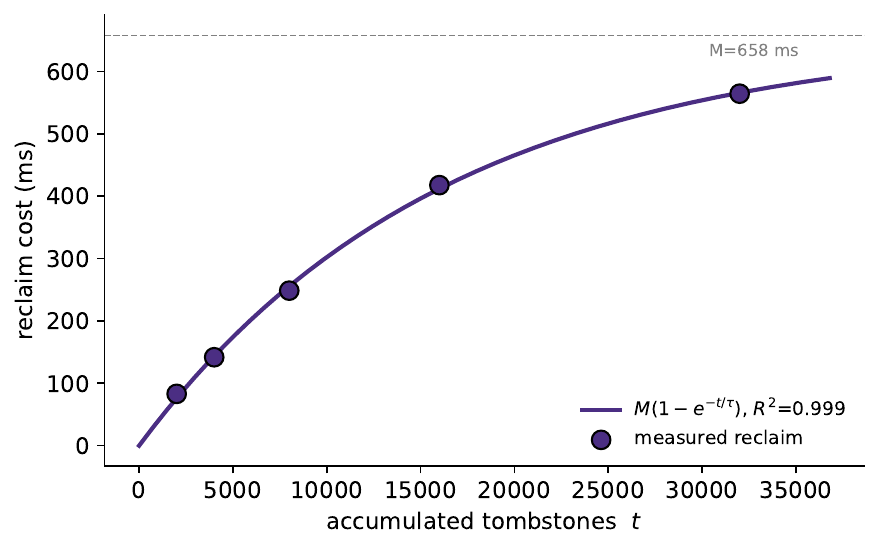}
    \caption{Measured reclaim cost with saturating fit ($R^2=0.999$).}
    \label{fig:reclaim-fit}
  \end{subfigure}\hfill
  \begin{subfigure}{0.49\linewidth}
    \includegraphics[width=\linewidth]{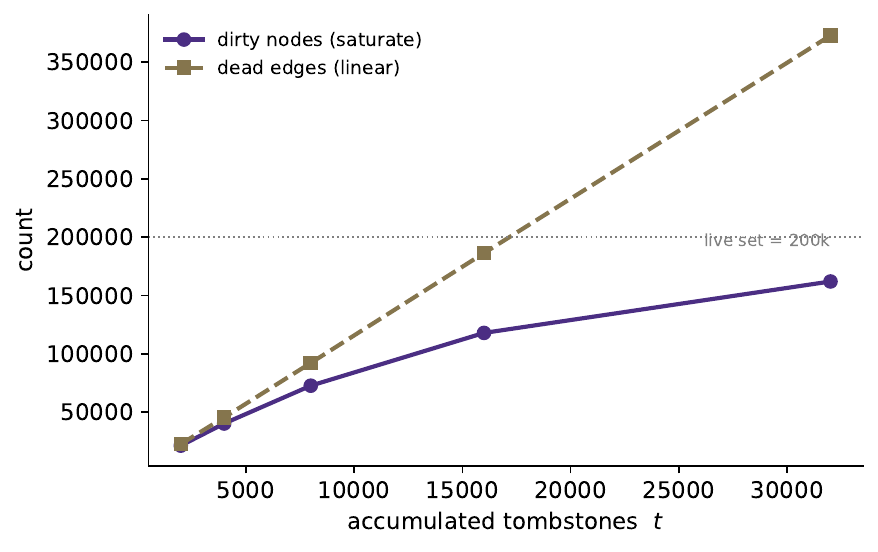}
    \caption{Dirty nodes saturate on live set; dead edges grow linearly.}
    \label{fig:reclaim-mech}
  \end{subfigure}
  \caption{Reclaim cost vs.\ accumulated tombstones.
  }
  \label{fig:reclaim}
\end{figure}

The ceiling $M$ is the cost of rebuilding the live set. In the limit where
every live node is dirty, reclaim degenerates to a full rebuild. The fit
recovers exactly this bound. The reclaim interval is in practice chosen well
below saturation, so each reclaim costs a small fraction of $M$ while keeping
the footprint bounded (Section~\ref{sec:eval-memory}).

\subsection{Recall Stability}
\label{sec:eval-stability}

A dynamic index is useful only if its accuracy does not erode as updates
accumulate. A lazy-deletion scheme leaves tombstones that distort the neighbor
graph. Recall falls round after round as the tombstones accumulate. ETALE
avoids this decay because its connectivity-repairing deletion and periodic
reclaim keep the maintained graph faithful. 
We expect that Recall@10 of ETALE would stay flat across
churn rounds.

\begin{figure}[t]
  \centering
  \begin{subfigure}{0.49\linewidth}
    \includegraphics[width=\linewidth]{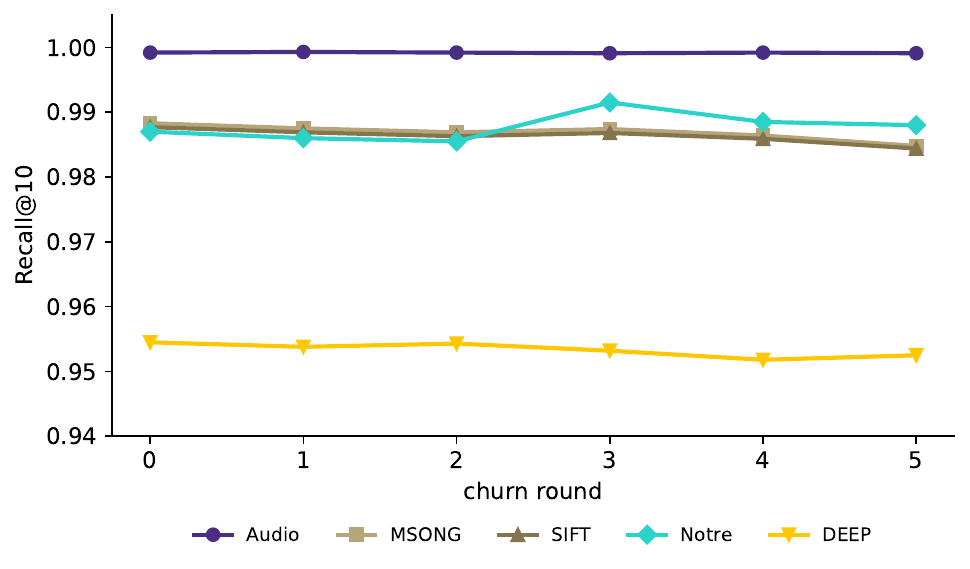}
    \caption{Across five datasets (fixed $200$k live set, $1\%$ churn/round).}
    \label{fig:stab-datasets}
  \end{subfigure}\hfill
  \begin{subfigure}{0.49\linewidth}
    \includegraphics[width=\linewidth]{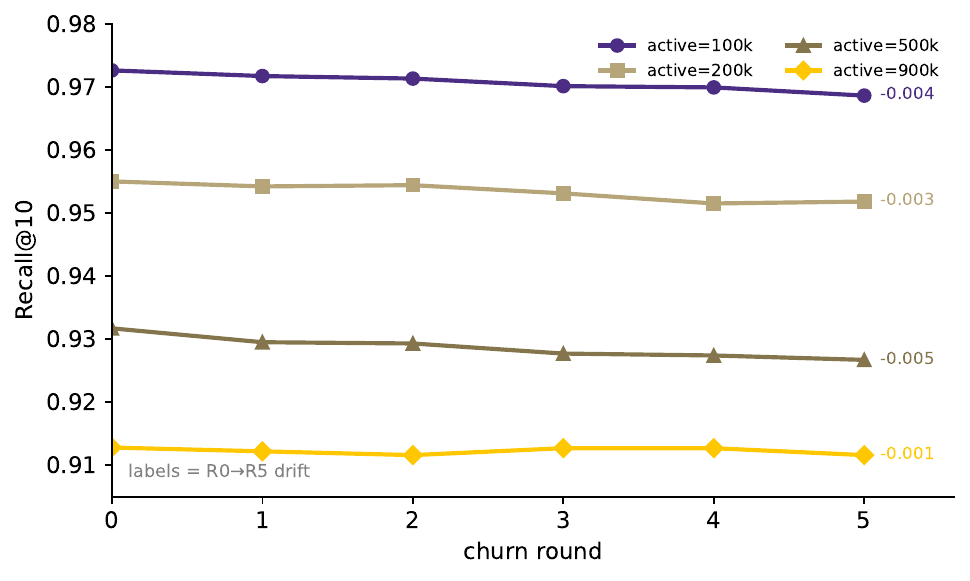}
    \caption{Across index scales on DEEP (active $100$k--$900$k).}
    \label{fig:stab-scale}
  \end{subfigure}
  \caption{Recall@10 over churn rounds stays flat as inserts and deletes
    accumulate, across both datasets and index sizes.
  }
  \label{fig:stability}
\end{figure}

Figure~\ref{fig:stab-datasets} tracks Recall@10 over five churn rounds on the
five datasets at a fixed $200$k live set. Every curve is flat with the largest
drift over five rounds at $-0.0035$ on MSONG and recall on Notre essentially
unchanged at $+0.0010$. These five datasets span image and audio
modalities across a range of dimensions, which demonstrates ETALE's stability.

Figure~\ref{fig:stab-scale} repeats the measurement on DEEP at four index sizes
spanning nearly an order of magnitude, with active sets of $100$k, $200$k,
$500$k, and $900$k under $1\%$ churn. Within each scale the curve is again
flat, with round-0-to-round-5 drift at most $-0.005$ at $500$k and as little as
$-0.0012$ at the largest size. The scales sit at different absolute recall
because a fixed query beam covers a larger graph less completely. This offset
is a static search-effort effect that does not reflect maintenance decay.

\subsection{Parameter Sensitivity}
\label{sec:eval-sensitivity}

The results so far fix the query beam width and the churn rate at single
values. This section varies each in turn, the query beam width that sets the
recall a query reaches and the churn rate that sets the volume of each update
round.
  
\begin{figure*}[t]
  \centering
  \begin{subfigure}{0.32\linewidth}
    \includegraphics[width=\linewidth]{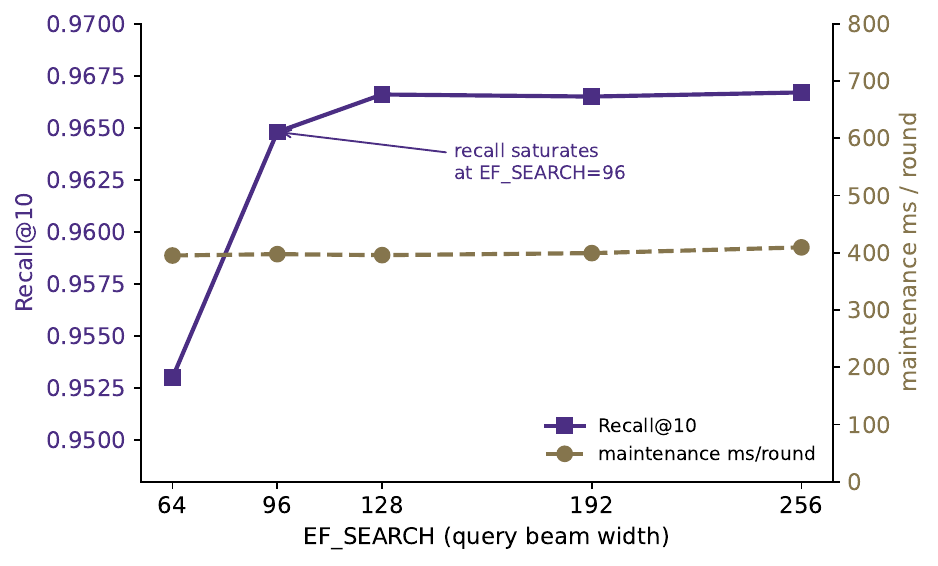}
    \caption{$\mathrm{EF\_SEARCH}$ sweep on DEEP. Recall@10 rises to a
      ceiling at $\mathrm{EF\_SEARCH}=96$ while maintenance cost stays flat.}
    \label{fig:ef-search}
  \end{subfigure}\hfill
  \begin{subfigure}{0.32\linewidth}
    \includegraphics[width=\linewidth]{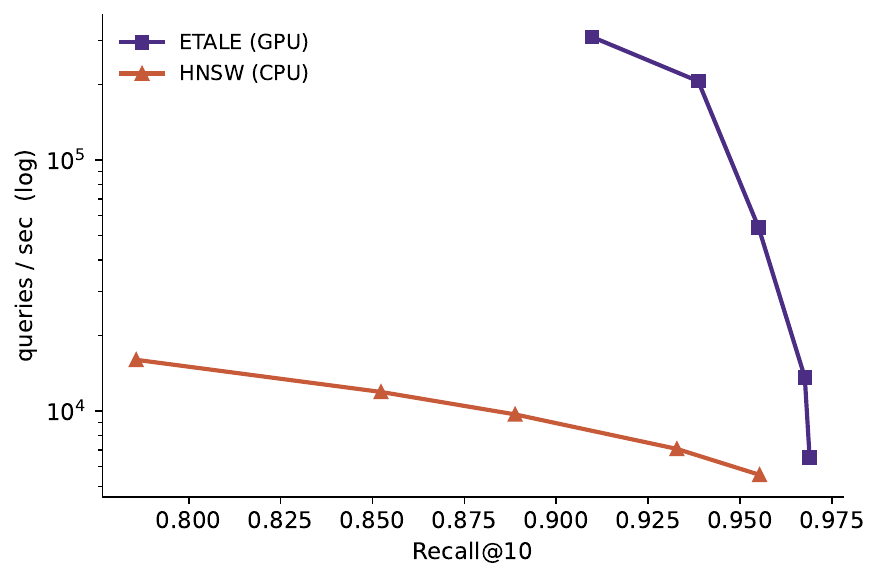}
    \caption{Recall--throughput on DEEP, ETALE (GPU) against HNSW (CPU),
      each sweeping its query beam. Throughput axis is log-scaled.}
    \label{fig:qps}
  \end{subfigure}\hfill
  \begin{subfigure}{0.32\linewidth}
    \includegraphics[width=\linewidth]{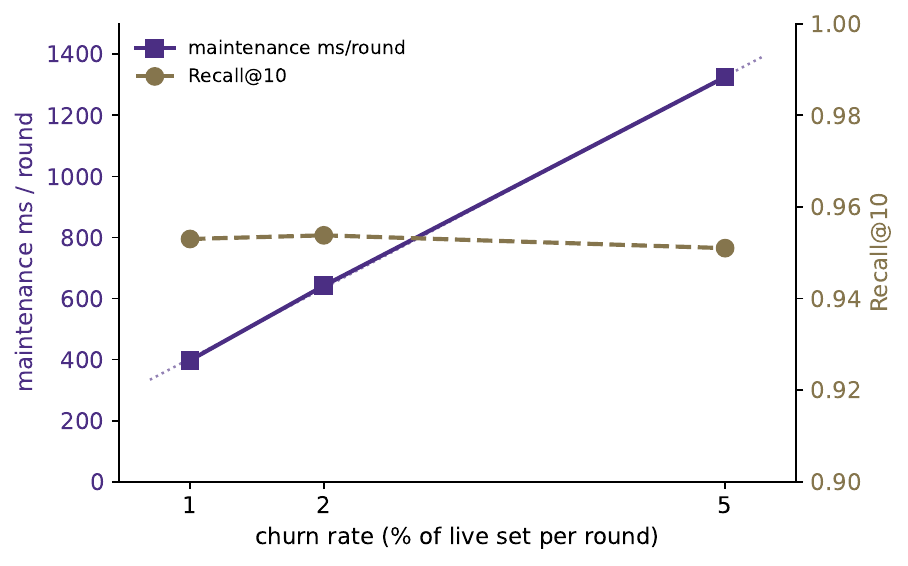}
    \caption{Churn-rate sweep on DEEP. Maintenance grows near-linearly
      ($398\to1325$\,ms over $1\%\to5\%$) while Recall@10 stays flat.}
    \label{fig:churn-rate}
  \end{subfigure}
  \caption{Parameter sensitivity of ETALE on DEEP.}
  \label{fig:sensitivity}
\end{figure*}

The query beam width is a free knob because ETALE decouples it from graph
construction. The build-time beam width $\mathrm{EF\_BUILD}$ fixes graph quality
once, whereas the query-time beam width $\mathrm{EF\_SEARCH}$ is paid on every
query and chosen independently of it. Recall is thus a query-time choice rather
than a property baked into the graph, and ETALE's fast maintenance leaves ample
headroom to widen the beam.
Figure~\ref{fig:ef-search} sweeps $\mathrm{EF\_SEARCH}$ on DEEP at fixed
$\mathrm{EF\_BUILD}$. Recall rises from $0.91$ at $\mathrm{EF\_SEARCH}=32$ to a
ceiling near $0.969$ at $\mathrm{EF\_SEARCH}=96$, while per-query maintenance
cost stays close to $400$\,ms throughout. The same operating points trace the
recall--throughput curve of Figure~\ref{fig:qps}, every one served by the same
maintained index, so accuracy against throughput is chosen at query time
without rebuilding.

Figure~\ref{fig:qps} places that curve against HNSW, the dynamic index ETALE
competes with on the query side, over the recall range where both operate as
practical search engines. ETALE answers queries at $5$--$20\times$ the
throughput of HNSW at matched recall, since its search runs on the GPU while
HNSW traverses the graph on the host. At Recall@10 of $0.955$, ETALE serves
about $5.4{\times}10^{4}$ queries/s against HNSW's $5.6{\times}10^{3}$. The
static GPU index CAGRA is not drawn here. It maintains no graph online and
is not a dynamic-index query peer; its comparison with ETALE is on maintenance
(Section~\ref{sec:eval-gpu}). ETALE is in this sense a dynamic-first index,
spending the static query specialization that a build-once index can pursue in
order to keep the graph mutable on the GPU, while still answering queries faster
than the dynamic baseline that shares its workload.

The churn rate sets how much each round changes. Figure~\ref{fig:churn-rate}
sweeps it over $1$, $2$, and $5\%$ of the live set per round on DEEP. Per-round
maintenance grows near-linearly with the churn volume, from $398$\,ms at $1\%$
to $1325$\,ms at $5\%$, while Recall@10 stays flat at $0.951$--$0.954$.
Maintenance cost is thus predictable during update.

\subsection{Ablation Studies}
\label{sec:eval-ablation}

The bounded footprint of Section~\ref{sec:eval-memory} is the product of the
periodic reclaim. This experiment isolates what that reclaim costs and what it
buys. We run the balanced churn workload on DEEP for five rounds with reclaim
enabled and disabled, holding every other setting fixed.
Table~\ref{tab:ablation} reports the resident slab count, the per-round reclaim
time, and recall.

\begin{table}[t]
  \centering
  \caption{Reclamation ablation under $1\%$ balanced churn.
  }
  \label{tab:ablation}
  \setlength{\tabcolsep}{6pt}
  \begin{tabular}{l c c}
    \toprule
    Metric                   & reclaim on       & reclaim off      \\
    \midrule
    Footprint after 5 rounds & $200$k (live)    & $210$k (growing) \\
    Footprint trend          & bounded          & unbounded        \\
    Reclaim time / round     & ${\sim}1150$\,ms & $0$              \\
    Maintenance / round      & $592$\,ms        & $579$\,ms        \\
    Recall@10                & $0.951$          & $0.953$          \\
    \bottomrule
  \end{tabular}
\end{table}

Reclaim does not change recall much, since it frees the slabs of already-tombstoned
nodes and never touches the live graph that search traverses. 
However, reclaim does impact the memory footprint. 
With reclaim on, the resident slab count stays pinned at the
live-set size of $200$k across all five rounds. With reclaim off, every round
adds its churn volume to the footprint without ever returning it, so the count
climbs linearly to $210$k after five rounds and would grow without bound over a
longer run. The cost of holding the footprint at the live-set size is therefore
one reclaim pass per round, on DEEP about $1150$\,ms, paid against a
maintenance step of comparable magnitude. This cost is itself bounded and
sublinear in accumulated tombstones, as Section~\ref{sec:eval-reclaim}
establishes. The reclaim interval then trades peak footprint against amortized
reclaim cost without either growing unbounded.

\section{Conclusion}
\label{sec:conclusion}

Dynamic ANN graphs on the GPU have lacked an index that maintains online and
whose deletion is provably monotone and whose footprint is provably bounded.
ETALE fills this gap with a lock-free copy-on-write slab structure, whose
deletion bit is inlined into the single word that every publication swaps by
compare-and-swap. 
On this structure, we prove deletion monotonicity, reclaim
safety, and a footprint bounded by the live-set size, and we establish a
reclaim cost that, under a uniform-placement model, is sublinear in accumulated
tombstones and ceilinged by the cost of rebuilding the live set. The
implementation maintains the index $4.8$--$8.8\times$ faster than CAGRA and
$1.8$--$2.5\times$ faster than the more recent Tagore at their per-window
rebuild, orders of magnitude faster than CPU dynamic indexes, at matched recall
and with a footprint that stays bounded where tombstone-only systems grow
indefinitely.


\section*{Disclosure of AI Usage}
System implementation and manuscript preparation were assisted by Copilot and Claude.
However, all the key components of this work, including but not limited to: design of data structures and algorithms, theoretical analysis, evaluation methods, were developed by the authors independently.

\begin{acks}
  Results presented in this paper were obtained using the Chameleon testbed supported by the National Science Foundation.
\end{acks}

\bibliographystyle{ACM-Reference-Format}
\bibliography{refs}

@inproceedings{dzhao_hpdc26,
  author    = {Zhao, Dongfang},
  title     = {{SIVF}: {GPU}-Resident {IVF} Index for Streaming Vector Analytics},
  booktitle = {Proceedings of the 35th International Symposium on High-Performance Parallel and Distributed Computing (HPDC '26)},
  year      = {2026},
  location  = {Cleveland, OH, USA},
  publisher = {Association for Computing Machinery},
  address   = {New York, NY, USA},
  pages     = {1--14},
  doi       = {10.1145/3806645.3807575},
  isbn      = {979-8-4007-2640-8},
}

@inproceedings{jwang_sigmod21,
author = {Wang, Jianguo and Yi, Xiaomeng and Guo, Rentong and Jin, Hai and Xu, Peng and Li, Shengjun and Wang, Xiangyu and Guo, Xiangzhou and Li, Chengming and Xu, Xiaohai and Yu, Kun and Yuan, Yuxing and Zou, Yinghao and Long, Jiquan and Cai, Yudong and Li, Zhenxiang and Zhang, Zhifeng and Mo, Yihua and Gu, Jun and Jiang, Ruiyi and Wei, Yi and Xie, Charles},
title = {Milvus: A Purpose-Built Vector Data Management System},
year = {2021},
isbn = {9781450383431},
publisher = {Association for Computing Machinery},
address = {New York, NY, USA},
url = {https://doi.org/10.1145/3448016.3457550},
doi = {10.1145/3448016.3457550},
booktitle = {Proceedings of the 2021 International Conference on Management of Data},
pages = {2614–2627},
numpages = {14},
location = {Virtual Event, China},
series = {SIGMOD '21}
}

@inproceedings{plewi_neurips20,
author = {Lewis, Patrick and Perez, Ethan and Piktus, Aleksandra and Petroni, Fabio and Karpukhin, Vladimir and Goyal, Naman and K\"{u}ttler, Heinrich and Lewis, Mike and Yih, Wen-tau and Rockt\"{a}schel, Tim and Riedel, Sebastian and Kiela, Douwe},
title = {Retrieval-augmented generation for knowledge-intensive NLP tasks},
year = {2020},
isbn = {9781713829546},
publisher = {Curran Associates Inc.},
address = {Red Hook, NY, USA},
booktitle = {Proceedings of the 34th International Conference on Neural Information Processing Systems},
articleno = {793},
numpages = {16},
location = {Vancouver, BC, Canada},
series = {NIPS '20}
}

@inproceedings{rying_kdd18,
author = {Ying, Rex and He, Ruining and Chen, Kaifeng and Eksombatchai, Pong and Hamilton, William L. and Leskovec, Jure},
title = {Graph Convolutional Neural Networks for Web-Scale Recommender Systems},
year = {2018},
isbn = {9781450355520},
publisher = {Association for Computing Machinery},
address = {New York, NY, USA},
url = {https://doi.org/10.1145/3219819.3219890},
doi = {10.1145/3219819.3219890},
booktitle = {Proceedings of the 24th ACM SIGKDD International Conference on Knowledge Discovery \& Data Mining},
pages = {974–983},
numpages = {10},
location = {London, United Kingdom},
series = {KDD '18}
}

@inproceedings{gblel_spaa20,
author = {Blelloch, Guy E. and Anderson, Daniel and Dhulipala, Laxman},
title = {ParlayLib - A Toolkit for Parallel Algorithms on Shared-Memory Multicore Machines},
year = {2020},
isbn = {9781450369350},
publisher = {Association for Computing Machinery},
address = {New York, NY, USA},
url = {https://doi.org/10.1145/3350755.3400254},
doi = {10.1145/3350755.3400254},
booktitle = {Proceedings of the 32nd ACM Symposium on Parallelism in Algorithms and Architectures},
pages = {507–509},
numpages = {3},
location = {Virtual Event, USA},
series = {SPAA '20}
}

@article{jalgl_asplos15,
author = {Alglave, Jade and Batty, Mark and Donaldson, Alastair F. and Gopalakrishnan, Ganesh and Ketema, Jeroen and Poetzl, Daniel and Sorensen, Tyler and Wickerson, John},
title = {GPU Concurrency: Weak Behaviours and Programming Assumptions},
year = {2015},
issue_date = {April 2015},
publisher = {Association for Computing Machinery},
address = {New York, NY, USA},
volume = {50},
number = {4},
issn = {0362-1340},
url = {https://doi.org/10.1145/2775054.2694391},
doi = {10.1145/2775054.2694391},
journal = {SIGPLAN Not.},
month = mar,
pages = {577–591},
numpages = {15}
}

@inproceedings{ysun_cikm24,
author = {Sun, Yiping and Shi, Yang and Du, Jiaolong},
title = {A Real-Time Adaptive Multi-Stream GPU System For Online Approximate Nearest Neighborhood Search},
year = {2024},
isbn = {9798400704369},
publisher = {Association for Computing Machinery},
address = {New York, NY, USA},
url = {https://doi.org/10.1145/3627673.3680054},
doi = {10.1145/3627673.3680054},
booktitle = {Proceedings of the 33rd ACM International Conference on Information and Knowledge Management},
pages = {4906–4913},
numpages = {8},
location = {Boise, ID, USA},
series = {CIKM '24}
}

@article{yzhu_sigmod24,
author = {Zhu, Yifan and Ma, Ruiyao and Zheng, Baihua and Ke, Xiangyu and Chen, Lu and Gao, Yunjun},
title = {GTS: GPU-based Tree Index for Fast Similarity Search},
year = {2024},
issue_date = {June 2024},
publisher = {Association for Computing Machinery},
address = {New York, NY, USA},
volume = {2},
number = {3},
url = {https://doi.org/10.1145/3654945},
doi = {10.1145/3654945},
journal = {Proc. ACM Manag. Data},
month = may,
articleno = {142},
numpages = {27}
}

@inproceedings{hwang_cikm21,
author = {Wang, Hui and Zhao, Wan-Lei and Zeng, Xiangxiang and Yang, Jianye},
title = {Fast k-NN Graph Construction by GPU based NN-Descent},
year = {2021},
isbn = {9781450384469},
publisher = {Association for Computing Machinery},
address = {New York, NY, USA},
url = {https://doi.org/10.1145/3459637.3482344},
doi = {10.1145/3459637.3482344},
booktitle = {Proceedings of the 30th ACM International Conference on Information \& Knowledge Management},
pages = {1929–1938},
numpages = {10},
location = {Virtual Event, Queensland, Australia},
series = {CIKM '21}
}

@article{zli_sigmod25,
author = {Li, Zhonggen and Ke, Xiangyu and Zhu, Yifan and Yu, Bocheng and Zheng, Baihua and Gao, Yunjun},
title = {Scalable Graph Indexing using GPUs for Approximate Nearest Neighbor Search},
year = {2025},
issue_date = {December 2025},
publisher = {Association for Computing Machinery},
address = {New York, NY, USA},
volume = {3},
number = {6},
url = {https://doi.org/10.1145/3769825},
doi = {10.1145/3769825},
journal = {Proc. ACM Manag. Data},
month = dec,
articleno = {360},
numpages = {27}
}

@ARTICLE{kvenk_tbd25,
  author={Venkatasubba, Karthik and Khan, Saim and Singh, Somesh and Simhadri, Harsha Vardhan and Vedurada, Jyothi},
  journal={IEEE Transactions on Big Data}, 
  title={BANG: Billion-Scale Approximate Nearest Neighbour Search Using a Single GPU}, 
  year={2025},
  volume={11},
  number={6},
  pages={3142-3157},
  doi={10.1109/TBDATA.2025.3581085}}

@INPROCEEDINGS{yyu_icde22,
author = { Yu, Yuanhang and Wen, Dong and Zhang, Ying and Qin, Lu and Zhang, Wenjie and Lin, Xuemin },
booktitle = { 2022 IEEE 38th International Conference on Data Engineering (ICDE) },
title = {{ GPU-accelerated Proximity Graph Approximate Nearest Neighbor Search and Construction }},
year = {2022},
volume = {},
ISSN = {},
pages = {552-564},
doi = {10.1109/ICDE53745.2022.00046},
url = {https://doi.ieeecomputersociety.org/10.1109/ICDE53745.2022.00046},
publisher = {IEEE Computer Society},
address = {Los Alamitos, CA, USA},
month =May}

@inproceedings{wzhao_icde20,
  author       = {Weijie Zhao and
                  Shulong Tan and
                  Ping Li},
  title        = {{SONG:} Approximate Nearest Neighbor Search on {GPU}},
  booktitle    = {36th {IEEE} International Conference on Data Engineering, {ICDE} 2020,
                  Dallas, TX, USA, April 20-24, 2020},
  pages        = {1033--1044},
  publisher    = {{IEEE}},
  year         = {2020},
  url          = {https://doi.org/10.1109/ICDE48307.2020.00094},
  doi          = {10.1109/ICDE48307.2020.00094},
  bibsource    = {dblp computer science bibliography, https://dblp.org}
}

@article{mwang_sigmod24,
author = {Wang, Mengzhao and Xu, Weizhi and Yi, Xiaomeng and Wu, Songlin and Peng, Zhangyang and Ke, Xiangyu and Gao, Yunjun and Xu, Xiaoliang and Guo, Rentong and Xie, Charles},
title = {Starling: An I/O-Efficient Disk-Resident Graph Index Framework for High-Dimensional Vector Similarity Search on Data Segment},
year = {2024},
issue_date = {February 2024},
publisher = {Association for Computing Machinery},
address = {New York, NY, USA},
volume = {2},
number = {1},
url = {https://doi.org/10.1145/3639269},
doi = {10.1145/3639269},
journal = {Proc. ACM Manag. Data},
month = mar,
articleno = {14},
numpages = {27}
}

@article{cwei_vldb20,
author = {Wei, Chuangxian and Wu, Bin and Wang, Sheng and Lou, Renjie and Zhan, Chaoqun and Li, Feifei and Cai, Yuanzhe},
title = {AnalyticDB-V: a hybrid analytical engine towards query fusion for structured and unstructured data},
year = {2020},
issue_date = {August 2020},
publisher = {VLDB Endowment},
volume = {13},
number = {12},
issn = {2150-8097},
url = {https://doi.org/10.14778/3415478.3415541},
doi = {10.14778/3415478.3415541},
journal = {Proc. VLDB Endow.},
month = aug,
pages = {3152–3165},
numpages = {14}
}

@inproceedings{qzhan_osdi23,
author = {Qianxi Zhang and Shuotao Xu and Qi Chen and Guoxin Sui and Jiadong Xie and Zhizhen Cai and Yaoqi Chen and Yinxuan He and Yuqing Yang and Fan Yang and Mao Yang and Lidong Zhou},
title = {{VBASE}: Unifying Online Vector Similarity Search and Relational Queries via Relaxed Monotonicity},
booktitle = {17th USENIX Symposium on Operating Systems Design and Implementation (OSDI 23)},
year = {2023},
isbn = {978-1-939133-34-2},
address = {Boston, MA},
pages = {377--395},
url = {https://www.usenix.org/conference/osdi23/presentation/zhang-qianxi},
publisher = {USENIX Association},
month = jul
}

@article{lpate_sigmod24,
author = {Patel, Liana and Kraft, Peter and Guestrin, Carlos and Zaharia, Matei},
title = {ACORN: Performant and Predicate-Agnostic Search Over Vector Embeddings and Structured Data},
year = {2024},
issue_date = {June 2024},
publisher = {Association for Computing Machinery},
address = {New York, NY, USA},
volume = {2},
number = {3},
url = {https://doi.org/10.1145/3654923},
doi = {10.1145/3654923},
journal = {Proc. ACM Manag. Data},
month = may,
articleno = {120},
numpages = {27}
}

@inproceedings{sgoll_www23,
author = {Gollapudi, Siddharth and Karia, Neel and Sivashankar, Varun and Krishnaswamy, Ravishankar and Begwani, Nikit and Raz, Swapnil and Lin, Yiyong and Zhang, Yin and Mahapatro, Neelam and Srinivasan, Premkumar and Singh, Amit and Simhadri, Harsha Vardhan},
title = {Filtered-DiskANN: Graph Algorithms for Approximate Nearest Neighbor Search with Filters},
year = {2023},
isbn = {9781450394161},
publisher = {Association for Computing Machinery},
address = {New York, NY, USA},
url = {https://doi.org/10.1145/3543507.3583552},
doi = {10.1145/3543507.3583552},
booktitle = {Proceedings of the ACM Web Conference 2023},
pages = {3406–3416},
numpages = {11},
location = {Austin, TX, USA},
series = {WWW '23}
}

@INPROCEEDINGS{dbara_iccv23,
  author={Baranchuk, Dmitry and Douze, Matthijs and Upadhyay, Yash and Yalniz, I. Zeki},
  booktitle={2023 IEEE/CVF International Conference on Computer Vision (ICCV)}, 
  title={DeDrift: Robust Similarity Search under Content Drift}, 
  year={2023},
  volume={},
  number={},
  pages={10992-11001},
  doi={10.1109/ICCV51070.2023.01012}}

@inproceedings{jruan_kdd25,
author = {Ruan, Jiancheng and Chen, Tingyang and Yang, Renchi and Ke, Xiangyu and Gao, Yunjun},
title = {Empowering Graph-based Approximate Nearest Neighbor Search with Adaptive Awareness Capabilities},
year = {2025},
isbn = {9798400714542},
publisher = {Association for Computing Machinery},
address = {New York, NY, USA},
url = {https://doi.org/10.1145/3711896.3736930},
doi = {10.1145/3711896.3736930},
booktitle = {Proceedings of the 31st ACM SIGKDD Conference on Knowledge Discovery and Data Mining V.2},
pages = {2444–2454},
numpages = {11},
location = {Toronto ON, Canada},
series = {KDD '25}
}

@article{syu_vldb25,
author = {Yu, Song and Lin, Shengyuan and Gong, Shufeng and Xie, Yongqing and Liu, Ruicheng and Zhou, Yijie and Sun, Ji and Zhang, Yanfeng and Li, Guoliang and Yu, Ge},
title = {A Topology-Aware Localized Update Strategy for Graph-Based ANN Index},
year = {2025},
issue_date = {November 2025},
publisher = {VLDB Endowment},
volume = {19},
number = {3},
issn = {2150-8097},
url = {https://doi.org/10.14778/3778092.3778108},
doi = {10.14778/3778092.3778108},
journal = {Proc. VLDB Endow.},
month = nov,
pages = {495–508},
numpages = {14}
}

@article{dliu_vldb25,
author = {Liu, Dawei and Zheng, Bolong and Yue, Ziyang and Ruan, Fuhao and Zhou, Xiaofang and Jensen, Christian S.},
title = {Wolverine: Highly Efficient Monotonic Search Path Repair for Graph-Based ANN Index Updates},
year = {2025},
issue_date = {March 2025},
publisher = {VLDB Endowment},
volume = {18},
number = {7},
issn = {2150-8097},
url = {https://doi.org/10.14778/3734839.3734860},
doi = {10.14778/3734839.3734860},
journal = {Proc. VLDB Endow.},
month = mar,
pages = {2268–2280},
numpages = {13}
}

@inproceedings{ypan_bigdata23,
  author       = {Yu Pan and
                  Jianxin Sun and
                  Hongfeng Yu},
  editor       = {Jingrui He and
                  Themis Palpanas and
                  Xiaohua Hu and
                  Alfredo Cuzzocrea and
                  Dejing Dou and
                  Dominik Slezak and
                  Wei Wang and
                  Aleksandra Gruca and
                  Jerry Chun{-}Wei Lin and
                  Rakesh Agrawal},
  title        = {LM-DiskANN: Low Memory Footprint in Disk-Native Dynamic Graph-Based
                  {ANN} Indexing},
  booktitle    = {{IEEE} International Conference on Big Data, BigData 2023, Sorrento,
                  Italy, December 15-18, 2023},
  pages        = {5987--5996},
  publisher    = {{IEEE}},
  year         = {2023},
  url          = {https://doi.org/10.1109/BigData59044.2023.10386517},
  doi          = {10.1109/BIGDATA59044.2023.10386517},
  bibsource    = {dblp computer science bibliography, https://dblp.org}
}

@misc{hxu_arxiv25,
      title={In-Place Updates of a Graph Index for Streaming Approximate Nearest Neighbor Search}, 
      author={Haike Xu and Magdalen Dobson Manohar and Philip A. Bernstein and Badrish Chandramouli and Richard Wen and Harsha Vardhan Simhadri},
      year={2025},
      eprint={2502.13826},
      archivePrefix={arXiv},
      primaryClass={cs.IR},
      url={https://arxiv.org/abs/2502.13826}, 
}

@inproceedings{yxu_sosp23,
author = {Xu, Yuming and Liang, Hengyu and Li, Jin and Xu, Shuotao and Chen, Qi and Zhang, Qianxi and Li, Cheng and Yang, Ziyue and Yang, Fan and Yang, Yuqing and Cheng, Peng and Yang, Mao},
title = {SPFresh: Incremental In-Place Update for Billion-Scale Vector Search},
year = {2023},
isbn = {9798400702297},
publisher = {Association for Computing Machinery},
address = {New York, NY, USA},
url = {https://doi.org/10.1145/3600006.3613166},
doi = {10.1145/3600006.3613166},
booktitle = {Proceedings of the 29th Symposium on Operating Systems Principles},
pages = {545–561},
numpages = {17},
location = {Koblenz, Germany},
series = {SOSP '23}
}

@misc{asing_arxiv21,
      title={FreshDiskANN: A Fast and Accurate Graph-Based ANN Index for Streaming Similarity Search}, 
      author={Aditi Singh and Suhas Jayaram Subramanya and Ravishankar Krishnaswamy and Harsha Vardhan Simhadri},
      year={2021},
      eprint={2105.09613},
      archivePrefix={arXiv},
      primaryClass={cs.IR},
      url={https://arxiv.org/abs/2105.09613}, 
}

@misc{mdouz_arxiv24,
      title={The Faiss library}, 
      author={Matthijs Douze and Alexandr Guzhva and Chengqi Deng and Jeff Johnson and Gergely Szilvasy and Pierre-Emmanuel Mazaré and Maria Lomeli and Lucas Hosseini and Hervé Jégou},
      year={2025},
      eprint={2401.08281},
      archivePrefix={arXiv},
      primaryClass={cs.LG},
      url={https://arxiv.org/abs/2401.08281}, 
}

@InProceedings{jmart_eccv18,
author = {Martinez, Julieta and Zakhmi, Shobhit and Hoos, Holger H. and Little, James J.},
title = {LSQ++: Lower running time and higher recall in multi-codebook quantization},
booktitle = {Proceedings of the European Conference on Computer Vision (ECCV)},
month = {September},
year = {2018}
}

@InProceedings{rguo_icml20,
  title = 	 {Accelerating Large-Scale Inference with Anisotropic Vector Quantization},
  author =       {Guo, Ruiqi and Sun, Philip and Lindgren, Erik and Geng, Quan and Simcha, David and Chern, Felix and Kumar, Sanjiv},
  booktitle = 	 {Proceedings of the 37th International Conference on Machine Learning},
  pages = 	 {3887--3896},
  year = 	 {2020},
  editor = 	 {III, Hal Daumé and Singh, Aarti},
  volume = 	 {119},
  series = 	 {Proceedings of Machine Learning Research},
  month = 	 {13--18 Jul},
  publisher =    {PMLR},
  url = 	 {https://proceedings.mlr.press/v119/guo20h.html},
}

@article{jgao_sigmod24,
author = {Gao, Jianyang and Long, Cheng},
title = {RaBitQ: Quantizing High-Dimensional Vectors with a Theoretical Error Bound for Approximate Nearest Neighbor Search},
year = {2024},
issue_date = {June 2024},
publisher = {Association for Computing Machinery},
address = {New York, NY, USA},
volume = {2},
number = {3},
url = {https://doi.org/10.1145/3654970},
doi = {10.1145/3654970},
journal = {Proc. ACM Manag. Data},
month = may,
articleno = {167},
numpages = {27}
}

@article{tge_tpami13,
author = {Ge, Tiezheng and He, Kaiming and Ke, Qifa and Sun, Jian},
title = {Optimized Product Quantization},
year = {2014},
issue_date = {April 2014},
publisher = {IEEE Computer Society},
address = {USA},
volume = {36},
number = {4},
issn = {0162-8828},
url = {https://doi.org/10.1109/TPAMI.2013.240},
doi = {10.1109/TPAMI.2013.240},
journal = {IEEE Trans. Pattern Anal. Mach. Intell.},
month = apr,
pages = {744–755},
numpages = {12}
}

@ARTICLE{mmuja_tpami14,
author={Muja, Marius and Lowe, David G.},
journal={ IEEE Transactions on Pattern Analysis \& Machine Intelligence },
title={{ Scalable Nearest Neighbor Algorithms for High Dimensional Data }},
year={2014},
volume={36},
number={11},
ISSN={1939-3539},
pages={2227-2240},
doi={10.1109/TPAMI.2014.2321376},
url = {https://doi.ieeecomputersociety.org/10.1109/TPAMI.2014.2321376},
publisher={IEEE Computer Society},
address={Los Alamitos, CA, USA},
month=nov}

@article{jbent_cacm75,
author = {Bentley, Jon Louis},
title = {Multidimensional binary search trees used for associative searching},
year = {1975},
issue_date = {Sept. 1975},
publisher = {Association for Computing Machinery},
address = {New York, NY, USA},
volume = {18},
number = {9},
issn = {0001-0782},
url = {https://doi.org/10.1145/361002.361007},
doi = {10.1145/361002.361007},
journal = {Commun. ACM},
month = sep,
pages = {509–517},
numpages = {9}
}

@inproceedings{aando_neurips15,
 author = {Andoni, Alexandr and Indyk, Piotr and Laarhoven, Thijs and Razenshteyn, Ilya and Schmidt, Ludwig},
 booktitle = {Advances in Neural Information Processing Systems},
 editor = {C. Cortes and N. Lawrence and D. Lee and M. Sugiyama and R. Garnett},
 pages = {},
 publisher = {Curran Associates, Inc.},
 title = {Practical and Optimal LSH for Angular Distance},
 url = {https://proceedings.neurips.cc/paper_files/paper/2015/file/2823f4797102ce1a1aec05359cc16dd9-Paper.pdf},
 volume = {28},
 year = {2015}
}

@inproceedings{qlv_vldb07,
author = {Lv, Qin and Josephson, William and Wang, Zhe and Charikar, Moses and Li, Kai},
title = {Multi-probe LSH: efficient indexing for high-dimensional similarity search},
year = {2007},
isbn = {9781595936493},
publisher = {VLDB Endowment},
booktitle = {Proceedings of the 33rd International Conference on Very Large Data Bases},
pages = {950–961},
numpages = {12},
location = {Vienna, Austria},
series = {VLDB '07}
}

@inproceedings{mdata_socg04,
author = {Datar, Mayur and Immorlica, Nicole and Indyk, Piotr and Mirrokni, Vahab S.},
title = {Locality-sensitive hashing scheme based on p-stable distributions},
year = {2004},
isbn = {1581138857},
publisher = {Association for Computing Machinery},
address = {New York, NY, USA},
url = {https://doi.org/10.1145/997817.997857},
doi = {10.1145/997817.997857},
booktitle = {Proceedings of the Twentieth Annual Symposium on Computational Geometry},
pages = {253–262},
numpages = {10},
location = {Brooklyn, New York, USA},
series = {SCG '04}
}

@inproceedings{pindy_stoc98,
author = {Indyk, Piotr and Motwani, Rajeev},
title = {Approximate nearest neighbors: towards removing the curse of dimensionality},
year = {1998},
isbn = {0897919629},
publisher = {Association for Computing Machinery},
address = {New York, NY, USA},
url = {https://doi.org/10.1145/276698.276876},
doi = {10.1145/276698.276876},
booktitle = {Proceedings of the Thirtieth Annual ACM Symposium on Theory of Computing},
pages = {604–613},
numpages = {10},
location = {Dallas, Texas, USA},
series = {STOC '98}
}

@article{mwang_vldb21,
author = {Wang, Mengzhao and Xu, Xiaoliang and Yue, Qiang and Wang, Yuxiang},
title = {A comprehensive survey and experimental comparison of graph-based approximate nearest neighbor search},
year = {2021},
issue_date = {July 2021},
publisher = {VLDB Endowment},
volume = {14},
number = {11},
issn = {2150-8097},
url = {https://doi.org/10.14778/3476249.3476255},
doi = {10.14778/3476249.3476255},
journal = {Proc. VLDB Endow.},
month = jul,
pages = {1964–1978},
numpages = {15}
}

@article{rqiu_sigmod25,
author = {Qiu, Runwen and Tang, Jing},
title = {Efficient Approximate Nearest Neighbor Search via Hemi-Sphere Centroids Graph},
year = {2025},
issue_date = {December 2025},
publisher = {Association for Computing Machinery},
address = {New York, NY, USA},
volume = {3},
number = {6},
url = {https://doi.org/10.1145/3769786},
doi = {10.1145/3769786},
journal = {Proc. ACM Manag. Data},
month = dec,
articleno = {321},
numpages = {26}
}

@article{xzhao_vldb23,
author = {Zhao, Xi and Tian, Yao and Huang, Kai and Zheng, Bolong and Zhou, Xiaofang},
title = {Towards Efficient Index Construction and Approximate Nearest Neighbor Search in High-Dimensional Spaces},
year = {2023},
issue_date = {April 2023},
publisher = {VLDB Endowment},
volume = {16},
number = {8},
issn = {2150-8097},
url = {https://doi.org/10.14778/3594512.3594527},
doi = {10.14778/3594512.3594527},
journal = {Proc. VLDB Endow.},
month = apr,
pages = {1979–1991},
numpages = {13}
}

@article{syang_vldb25,
author = {Yang, Shuo and Xie, Jiadong and Liu, Yingfan and Yu, Jeffrey Xu and Gao, Xiyue and Wang, Qianru and Peng, Yanguo and Cui, Jiangtao},
title = {Revisiting the Index Construction of Proximity Graph-Based Approximate Nearest Neighbor Search},
year = {2025},
issue_date = {February 2025},
publisher = {VLDB Endowment},
volume = {18},
number = {6},
issn = {2150-8097},
url = {https://doi.org/10.14778/3725688.3725709},
doi = {10.14778/3725688.3725709},
journal = {Proc. VLDB Endow.},
month = feb,
pages = {1825–1838},
numpages = {14}
}

@inproceedings{mmano_ppopp24,
author = {Manohar, Magdalen Dobson and Shen, Zheqi and Blelloch, Guy and Dhulipala, Laxman and Gu, Yan and Simhadri, Harsha Vardhan and Sun, Yihan},
title = {ParlayANN: Scalable and Deterministic Parallel Graph-Based Approximate Nearest Neighbor Search Algorithms},
year = {2024},
isbn = {9798400704352},
publisher = {Association for Computing Machinery},
address = {New York, NY, USA},
url = {https://doi.org/10.1145/3627535.3638475},
doi = {10.1145/3627535.3638475},
booktitle = {Proceedings of the 29th ACM SIGPLAN Annual Symposium on Principles and Practice of Parallel Programming},
pages = {270–285},
numpages = {16},
location = {Edinburgh, United Kingdom},
series = {PPoPP '24}
}

@inproceedings{nono_mm23,
author = {Ono, Naoki and Matsui, Yusuke},
title = {Relative NN-Descent: A Fast Index Construction for Graph-Based Approximate Nearest Neighbor Search},
year = {2023},
isbn = {9798400701085},
publisher = {Association for Computing Machinery},
address = {New York, NY, USA},
url = {https://doi.org/10.1145/3581783.3612290},
doi = {10.1145/3581783.3612290},
booktitle = {Proceedings of the 31st ACM International Conference on Multimedia},
pages = {1659–1667},
numpages = {9},
location = {Ottawa ON, Canada},
series = {MM '23}
}

@misc{cfu_arxiv16,
      title={EFANNA : An Extremely Fast Approximate Nearest Neighbor Search Algorithm Based on kNN Graph}, 
      author={Cong Fu and Deng Cai},
      year={2016},
      eprint={1609.07228},
      archivePrefix={arXiv},
      primaryClass={cs.CV},
      url={https://arxiv.org/abs/1609.07228}, 
}

@inproceedings{wdong_www11,
author = {Dong, Wei and Moses, Charikar and Li, Kai},
title = {Efficient k-nearest neighbor graph construction for generic similarity measures},
year = {2011},
isbn = {9781450306324},
publisher = {Association for Computing Machinery},
address = {New York, NY, USA},
url = {https://doi.org/10.1145/1963405.1963487},
doi = {10.1145/1963405.1963487},
booktitle = {Proceedings of the 20th International Conference on World Wide Web},
pages = {577–586},
numpages = {10},
location = {Hyderabad, India},
series = {WWW '11}
}

@article{jmuno_pr19,
author = {Vargas Mu\~{n}oz, Javier and Gon\c{c}alves, Marcos A. and Dias, Zanoni and da S. Torres, Ricardo},
title = {Hierarchical Clustering-Based Graphs for Large Scale Approximate Nearest Neighbor Search},
year = {2019},
issue_date = {Dec 2019},
publisher = {Elsevier Science Inc.},
address = {USA},
volume = {96},
number = {C},
issn = {0031-3203},
url = {https://doi.org/10.1016/j.patcog.2019.106970},
doi = {10.1016/j.patcog.2019.106970},
journal = {Pattern Recogn.},
month = dec,
numpages = {10}
}

@InProceedings{bharw_cvpr16,
author = {Harwood, Ben and Drummond, Tom},
title = {FANNG: Fast Approximate Nearest Neighbour Graphs},
booktitle = {Proceedings of the IEEE Conference on Computer Vision and Pattern Recognition (CVPR)},
month = {June},
year = {2016}
}

@article{cfu_vldb19,
author = {Fu, Cong and Xiang, Chao and Wang, Changxu and Cai, Deng},
title = {Fast approximate nearest neighbor search with the navigating spreading-out graph},
year = {2019},
issue_date = {January 2019},
publisher = {VLDB Endowment},
volume = {12},
number = {5},
issn = {2150-8097},
url = {https://doi.org/10.14778/3303753.3303754},
doi = {10.14778/3303753.3303754},
journal = {Proc. VLDB Endow.},
month = jan,
pages = {461–474},
numpages = {14}
}

@article{ymalk_is14,
title = {Approximate nearest neighbor algorithm based on navigable small world graphs},
journal = {Information Systems},
volume = {45},
pages = {61-68},
year = {2014},
issn = {0306-4379},
doi = {https://doi.org/10.1016/j.is.2013.10.006},
url = {https://www.sciencedirect.com/science/article/pii/S0306437913001300},
author = {Yury Malkov and Alexander Ponomarenko and Andrey Logvinov and Vladimir Krylov},
}

@ARTICLE{dpg_audio_notre,
  author={Li, Wen and Zhang, Ying and Sun, Yifang and Wang, Wei and Li, Mingjie and Zhang, Wenjie and Lin, Xuemin},
  journal={IEEE Transactions on Knowledge and Data Engineering}, 
  title={Approximate Nearest Neighbor Search on High Dimensional Data — Experiments, Analyses, and Improvement}, 
  year={2020},
  volume={32},
  number={8},
  pages={1475-1488},
  doi={10.1109/TKDE.2019.2909204}}

@inproceedings{msong,
author = {McFee, Brian and Bertin-Mahieux, Thierry and Ellis, Daniel P.W. and Lanckriet, Gert R.G.},
title = {The million song dataset challenge},
year = {2012},
isbn = {9781450312301},
publisher = {Association for Computing Machinery},
address = {New York, NY, USA},
url = {https://doi.org/10.1145/2187980.2188222},
doi = {10.1145/2187980.2188222},
booktitle = {Proceedings of the 21st International Conference on World Wide Web},
pages = {909–916},
numpages = {8},
location = {Lyon, France},
series = {WWW '12 Companion}
}

@INPROCEEDINGS{deep,
  author={Yandex, Artem Babenko and Lempitsky, Victor},
  booktitle={2016 IEEE Conference on Computer Vision and Pattern Recognition (CVPR)}, 
  title={Efficient Indexing of Billion-Scale Datasets of Deep Descriptors}, 
  year={2016},
  volume={},
  number={},
  pages={2055-2063},
  doi={10.1109/CVPR.2016.226}}

@ARTICLE{sift,
  author={Jégou, Herve and Douze, Matthijs and Schmid, Cordelia},
  journal={IEEE Transactions on Pattern Analysis and Machine Intelligence}, 
  title={Product Quantization for Nearest Neighbor Search}, 
  year={2011},
  volume={33},
  number={1},
  pages={117-128},
  doi={10.1109/TPAMI.2010.57}}

@article{dzhao_tpami09,
author = {Zhao, Dongfang and Yang, Li},
title = {Incremental Isometric Embedding of High-Dimensional Data Using Connected Neighborhood Graphs},
year = {2009},
issue_date = {January 2009},
publisher = {IEEE Computer Society},
address = {USA},
volume = {31},
number = {1},
issn = {0162-8828},
url = {https://doi.org/10.1109/TPAMI.2008.34},
doi = {10.1109/TPAMI.2008.34},
journal = {IEEE Trans. Pattern Anal. Mach. Intell.},
month = jan,
pages = {86–98},
numpages = {13}
}

@article{digra,
author = {Jiang, Mengxu and Yang, Zhi and Zhang, Fangyuan and Hou, Guanhao and Shi, Jieming and Zhou, Wenchao and Li, Feifei and Wang, Sibo},
title = {DIGRA: A Dynamic Graph Indexing for Approximate Nearest Neighbor Search with Range Filter},
year = {2025},
issue_date = {June 2025},
publisher = {Association for Computing Machinery},
address = {New York, NY, USA},
volume = {3},
number = {3},
url = {https://doi.org/10.1145/3725399},
doi = {10.1145/3725399},
journal = {Proc. ACM Manag. Data},
month = jun,
articleno = {148},
numpages = {26}
}

@INPROCEEDINGS {cagra,
author = { Ootomo, Hiroyuki and Naruse, Akira and Nolet, Corey and Wang, Ray and Feher, Tamas and Wang, Yong },
booktitle = { 2024 IEEE 40th International Conference on Data Engineering (ICDE) },
title = {{ CAGRA: Highly Parallel Graph Construction and Approximate Nearest Neighbor Search for GPUs }},
year = {2024},
volume = {},
ISSN = {},
pages = {4236-4247},
doi = {10.1109/ICDE60146.2024.00323},
url = {https://doi.ieeecomputersociety.org/10.1109/ICDE60146.2024.00323},
publisher = {IEEE Computer Society},
address = {Los Alamitos, CA, USA},
month =May}

@article{hnsw,
author = {Malkov, Yu A. and Yashunin, D. A.},
title = {Efficient and Robust Approximate Nearest Neighbor Search Using Hierarchical Navigable Small World Graphs},
year = {2020},
issue_date = {April 2020},
publisher = {IEEE Computer Society},
address = {USA},
volume = {42},
number = {4},
issn = {0162-8828},
url = {https://doi.org/10.1109/TPAMI.2018.2889473},
doi = {10.1109/TPAMI.2018.2889473},
journal = {IEEE Trans. Pattern Anal. Mach. Intell.},
month = apr,
pages = {824–836},
numpages = {13}
}

@inproceedings{diskann,
 author = {Jayaram Subramanya, Suhas and Devvrit, Fnu and Simhadri, Harsha Vardhan and Krishnawamy, Ravishankar and Kadekodi, Rohan},
 booktitle = {Advances in Neural Information Processing Systems},
 editor = {H. Wallach and H. Larochelle and A. Beygelzimer and F. d\textquotesingle Alch\'{e}-Buc and E. Fox and R. Garnett},
 pages = {},
 publisher = {Curran Associates, Inc.},
 title = {DiskANN: Fast Accurate Billion-point Nearest Neighbor Search on a Single Node},
 url = {https://proceedings.neurips.cc/paper_files/paper/2019/file/09853c7fb1d3f8ee67a61b6bf4a7f8e6-Paper.pdf},
 volume = {32},
 year = {2019}
}

\end{document}